\documentclass[journal,11pt,twoside]{IEEEtranTMBMC}
\normalsize

\usepackage{graphicx}%
\usepackage{multirow}%
\usepackage{tikz}
\usetikzlibrary{arrows.meta, positioning, shapes.geometric, fit, calc, decorations.pathreplacing, backgrounds}
\usepackage{pgfplots}
\usepackage{amsmath,amssymb,amsfonts}%
\usepackage{amsthm}%
\usepackage{mathrsfs}%
\usepackage{hyperref}
\usepackage{cleveref}
\usepackage{xcolor}%
\usepackage{textcomp}%
\usepackage{manyfoot}%
\usepackage{booktabs}%
\usepackage{algorithm}%
\usepackage{algorithmicx}%
\usepackage{algpseudocode}%
\usepackage{listings}%
\usepackage[nolist]{acronym}%
\usepackage{mathtools,scalerel}
\usepackage{cite}
\usepackage{url}
\usepackage{anyfontsize}

\usepackage{paralist}

\colorlet{myteal}{teal!75!yellow}

\definecolor{bleudefrance}{rgb}{0.19, 0.55, 0.91}

\newcommand{\prob}[1]{\mathrm{{Pr}}\left[ #1\right]}

\newcommand{\wt}[1]{\mathrm{wt}_{\mathrm{H}}(#1)} 
\newcommand{\dist}[1]{\mathrm{d}_{\mathrm{H}}(#1)} 
\newcommand{\jaccard}[1]{\mathrm{J_w}(#1)} 
\newcommand{\F}{\mathbb{F}}
\newcommand{\Fq}{\F_q}
\newcommand{\bindis}[1]{\ensuremath{\mathrm{Bin}\!\left(#1\right)}} 
\newcommand{\bincdf}[1]{\ensuremath{F_{\mathrm{Bin}}\!\left(#1\right)}} 
\newcommand{\E}{\mathbb{E}} 

\newcommand{\set}[1]{\left\lbrace #1 \right\rbrace}
\newcommand{\card}[1]{\left| #1 \right|}
\newcommand{\st}{\; | \;}
\newcommand{\pro}{\mathsf{pro}}
\newcommand{\chal}{\mathsf{chal}}
\newcommand{\resp}{\mathsf{resp}}

\newcommand{\outp}{\mathsf{out}}
\newcommand{\helper}{\mathsf{help}}
\newcommand{\setprofile}{\mathcal{S}_{\pro}}
\newcommand{\setchal}{\mathcal{S}_{\chal}}
\newcommand{\setresp}{\mathcal{S}_{\resp}}
\newcommand{\setext}{\mathcal{S}_{\outp}}
\newcommand{\sethelper}{\mathcal{S}_{\helper}}

\newcommand{\para}[1]{\alpha_{\mathrm{#1}}}

\renewcommand{\to}{\longrightarrow}
\renewcommand{\mapsto}{\longmapsto}

\newcommand{\keygen}{\mathsf{KeyGen}}
\newcommand{\gen}{\mathsf{Gen}}

\newcommand{\mac}{\mathsf{MAC}}
\newcommand{\macforge}{\mathsf{MAC\text{-}forge}}
\newcommand{\mack}{\mathsf{MAC}_{k}}

\newcommand{\verify}{\mathsf{Verify}}
\newcommand{\eval}{\mathsf{Eval}}
\newcommand{\attack}{\mathsf{Clone}}
\newcommand{\openattack}{\mathsf{OpenClone}}
\newcommand{\sattack}{\mathsf{Predict}}
\newcommand{\negl}{\mathsf{negl}}

\newcommand{\extract}{\mathsf{Extract}}

\newcommand{\PF}[1]{\mathsf{PF}_{#1}}

\newcommand{\cf}[1]{\mathsf{CF}_{#1}} 
\newcommand{\cfs}[1]{\mathsf{CFS}_{#1}} 
\newcommand{\cfi}{\mathsf{CFI}} 

\newcommand{\profile}{\mathrm{p}} 
\newcommand{\physcomp}{\mathrm{p}}

\newcommand{\adversary}{\mathscr{A}}

\newcommand{\A}{\mathsf{A}}
\newcommand{\C}{\mathsf{C}}
\newcommand{\G}{\mathsf{G}}
\newcommand{\T}{\mathsf{T}}

\theoremstyle{plain}%
\newtheorem{theorem}{Theorem}
\newtheorem{proposition}[theorem]{Proposition}%
\newtheorem{corollary}[theorem]{Corollary}%

\theoremstyle{remark}%
\newtheorem{example}{Example}%

\theoremstyle{definition}%
\newtheorem{definition}{Definition}%

\pgfplotsset{compat=1.18}

\usepackage{fancyhdr}

\fancypagestyle{IEEEtitlepagestyle}{%
\fancyhf{}%
\cfoot{\raisebox{-10pt}[0pt][0pt]{\parbox[b]{\textwidth}{\centering\scriptsize\textcopyright 2026 IEEE. Personal use of this material is permitted. Permission from IEEE must be obtained for all other uses, in any current or future media, including reprinting/republishing this material for advertising or promotional purposes, creating new collective works, for resale or redistribution to servers or lists, or reuse of any copyrighted component of this work in other works.
DOI: \href{https://doi.org/10.1109/TMBMC.2026.3732978}{10.1109/TMBMC.2026.3732978}}}}}

\begin{document}
%
\title{A Security Framework for Chemical Functions}
%
%
%


\author{Frederik~Walter,~\IEEEmembership{Student Member,~IEEE,}
       Hrishi~Narayanan,  ~\IEEEmembership{Student Member,~IEEE,}
       Jessica~Bariffi,~\IEEEmembership{Member,~IEEE,}
       Anne~Lüscher,
       Rawad~Bitar,~\IEEEmembership{Member,~IEEE,}
       Robert~Grass,
       Antonia~Wachter-Zeh,~\IEEEmembership{Senior Member,~IEEE,}
       and Zohar~Yakhini,~\IEEEmembership{Member,~IEEE,}
\thanks{F. Walter, H. Narayanan, J. Bariffi, R. Bitar and A. Wachter-Zeh are with the Department of Computer Engineering,
                   Technical University of Munich, Germany.}
\thanks{A. Lüscher and R. Grass are with the Department of Chemistry and Applied Biosciences, ETH Zürich, Switzerland.}
\thanks{Z. Yakhini is with the School of Computer Science, Reichman University, Herzlia, Israel and with the Computer Science Faculty, Technion, Haifa, Israel.}
\thanks{Corresponding author: F. Walter (e-mail: frederik.walter@tum.de)}
}

\begin{acronym}
    \acro{abe}[ABE]{\emph{adenine base editor}}
    \acro{cbe}[CBE]{\emph{cytosine base editor}}
    \acro{vap}[VAP]{\emph{verifiable authentication protocol}}
    \acro{mac}[MACS]{\emph{Message Authentication Code System}}
    \acro{uprot}[UVAP]{\emph{unclonable verifiable authentication protocol}}
    \acro{crp}[CRP]{\emph{Challenge-Response Pair}}
    \acro{secure}[UD-CCA]{\emph{Unclonable DNA for Chosen Challenge Attacks}}
    \acro{ssecure}[UR-CCA]{\emph{Unpredictable Response for Chosen Challenge Attacks}}
    \acro{eufcma}[EUF-CMA]{\emph{existentially unforgeable under an adaptive chosen-message attack}}
    \acro{puf}[PUF]{\emph{Physically Unclonable Function}}
    \acro{pf}[PF]{\emph{Physical Function}}
    \acro{pfs}[PFS]{\emph{Physical Function System}}
    \acro{cf}[CF]{\emph{Chemical Function}}
    \acro{cfs}[CFS]{\emph{Chemical Function System}}
    \acro{cfi}[CFI]{\emph{Chemical Function Infrastructure}}
    \acro{cuf}[CUF]{\emph{Chemically Unclonable Function}}
    \acro{dna}[DNA]{\emph{Deoxyribonucleic acid}}
    \acro{pcr}[PCR]{\emph{Polymerase Chain Reaction}}
    \acro{nist}[NIST]{\emph{National Institute of Standards and Technology}}
    \acro{ordna}[orDNA]{\emph{operable random DNA}}
    \acro{gse}[GSE]{\emph{Genomic Sequence Encryption}}
\end{acronym}

\maketitle

\begin{abstract}
    In this paper, we provide a unified \emph{security framework} with the goal of being able to evaluate and compare different (known and new) approaches for authentication and encryption using, for example, DNA-based schemes.
Therefore, we introduce and rigorously define chemical functions which model chemical systems as noisy challenge--response primitives, and formalize the associated chemical function infrastructure. 
Building on the theory of physical functions, we rigorously define robustness, unclonability, and unpredictability for chemical functions in both finite and asymptotic regimes, and specify security games that capture the adversary's power and the security goals.
We instantiate the framework with the data from two existing DNA-based constructions (``operable random DNA'' and ``Genomic Sequence Encryption'') and derive quantitative bounds for robustness, unclonability, and unpredictability.

These results place DNA-based chemical functions on a rigorous cryptographic footing, enabling principled design and comparison of chemically grounded authentication mechanisms.
We demonstrate applications to in-product authentication and to shared key generation using standard extraction techniques, showing the significance of this work regarding real-world applications.

\end{abstract}


%
\IEEEpeerreviewmaketitle

\section{Introduction}\label{sec:introduction}

Across supply chains, healthcare, and critical infrastructure, there is an ever-present need to bind trust to the physical world. While cryptography provides rigorous guarantees for digital artifacts, ensuring the authenticity, integrity, and provenance of physical matter remains challenging. Modern taggants, labels, and device-bound fingerprints (e.g., \acp{puf}) offer valuable anchors but are typically confined to surfaces, specific components, or narrowly scoped measurement setups. This motivates security mechanisms whose guarantees arise from the behavior of chemical systems themselves—mechanisms that can be embedded throughout materials and verified via standard laboratory operations.
Some cryptographic schemes based on chemical systems have already been suggested in the past \cite{luescher_2024_ChemicalUnclonablea, volf_2023_CryptographyDNA, berezin_2024_CryptographicApproaches, heider_2007_DNAbasedWatermarks, heider_2008_DNAWatermarks, jupiter_2010_DNAWatermarking, lee_2014_DWTBased, kar_2018_DigitalSignatures, kar_2019_SystematizationKnowledge, kar_2019_ThatsMy, kar_2020_SynthesizingDNA, clelland_1999_HidingMessages, vippathalla_2023_SecureStorage, leier_2000_CryptographyDNA, cui_2019_IncorporatingRandomness, grass_2020_GenomicEncryption, schaudy_2020_LDNADuplex,gupta_2026_DNAData,limor_2026_StochasticPerformance}. 
However, no unified framework to compare these approaches was available to assess their security. 

In this paper, we provide a unifying view on the security of chemical systems based on chemical functions: mappings from chemical states and admissible operations to measurable outputs. In this view, the complexity originates from intrinsic physical or chemical processes. This generalizes hardware-centric notions such as \acp{puf} to settings where the function is performed by chemistry and where the verifier interacts by applying laboratory procedures under a specified protocol. \acp{cuf} as first introduced in \cite{luescher_2024_ChemicalUnclonablea} arise as an important subclass: they are designed so that reproducing an equivalent chemical object is infeasible even for powerful adversaries with laboratory access. 
We see the key difference between chemical functions and physical functions as the ability to be divided into smaller parts and distributed across space, whereas physical functions are usually bound to a specific device or component.

Despite rapid progress, previous work is often presented with domain-specific abstractions and incomparable metrics, making it difficult to assess fundamental trade-offs, reason about adversarial capabilities, or compose mechanisms. In particular, a lack of common terminology for (i) parties and roles, (ii) power of the adversary, (iii) security goals, and (iv) security notions (e.g., unclonability, unpredictability) hinders systematic analysis. 

This paper introduces a theoretical framework for chemical functions that abstracts chemical mechanisms into analyzable components and aligns them with cryptographic methodology. 
We formalize the components of chemical functions and the necessary infrastructure, adversary power, security goals, and specific properties of chemical functions. This paper presents a purely theoretical framework for evaluating the security of chemical authentication primitives, without introducing novel wet-lab primitives or including a large-scale experimental deployment. In particular, we do not provide benchmark measurements across different environments (e.g., laboratory versus field conditions), nor a head-to-head comparison with existing industrial authentication systems. The framework is intended as a tool to guide the design of chemical functions and the experiments needed to assess their security, rather than as a substitute for such experimental validation.
Concretely, our contributions are:
\begin{itemize}
    \item We define the chemical function abstraction that generalizes device-focused notions (e.g., \acp{puf}) to chemistry-driven divisible settings and captures \acp{cuf} as a special case.
    \item We define key properties of chemical functions relevant to their practical application in security applications. 
    \item We formulate security goals and adversary capabilities for unclonability and unpredictability. 
    \item We instantiate the framework with chemical authentication schemes and \ac{dna}-based mechanisms, including operable random sequence pools \cite{luescher_2024_ChemicalUnclonablea} and genomic sequence encryption \cite{volf_2023_CryptographyDNA}, showing how protocol choices and measurement noise map to security and reliability guarantees.
\end{itemize}

By decoupling chemical specifics from security reasoning, the framework enables principled design, analysis, and comparison of chemical functions. It clarifies how chemical operations act as cryptographic resources, how access policies constrain adversaries, and how verification can be made robust in realistic laboratory conditions.

This paper is organized as follows.
In \Cref{sec:related}, we introduce preliminary concepts such as message authentication schemes and classical security notions, and we discuss related works needed throughout the paper. Translating the terminology and concepts of \acp{puf}, we introduce chemical functions in \Cref{sec:chemical_functions}. In \Cref{sec:properties}, we develop properties and metrics relevant to the security and reliability of chemical functions in both a finite and an asymptotic domain. 
Practical limitations such as DNA amplification or lack of detailed cost analysis are discussed in \Cref{sec:limitations}.
We show how chemical functions can be used to realize authentication schemes and key generation protocols in \Cref{sec:application}. In \Cref{sec:genomic-sequence-encryption}, we reformulate the genomic sequence encryption scheme from \cite{volf_2023_CryptographyDNA} as a possible authentication scheme and apply our framework to that chemical function. Finally, \Cref{sec:conclusions} concludes the paper and outlines future directions.

\section{Preliminaries and Related Works}\label{sec:related}
In this section, we recap the main definitions and concepts that lay the mathematical and historical foundation needed throughout this paper. For the analysis of existing schemes, we will use $\Fq$ to denote a finite field of order $q$ where $q$ is a prime power. Additionally, for any length-$n$ vector $x \in \Fq^n$ we denote by $\wt{x}$ its \emph{Hamming weight}, that is $\wt{x} = \card{\set{i = 1, \ldots, n \st x_i \neq 0}}$. Similarly, for $x, y\in \Fq^n$, we use $\dist{x, y}$ to denote their \emph{Hamming distance}, i.e.,  the number of positions in which $x$ and $y$ differ.
\subsection{Authentication Schemes}\label{subsec:authentication}
Authentication schemes are protocols used to verify the identity of a party or the integrity of data. That is, a receiver of a message is able to verify the sender's identity (\emph{authentication}) as well as the integrity of the message to make sure it has not been altered (\emph{integrity}). 
In cryptography, the security of schemes is usually assessed in terms of asymptotic security as follows.
We require that such a scheme cannot be broken using a polynomial-time attack with ``relatively good probability'' of success. 
The polynomial running time, as well as the probability of success of the attack, are defined as functions of a previously fixed \emph{security parameter} $n$, which is publicly known, e.g., also to any adversary.
In other words, an authentication scheme is considered to be secure if no polynomial-time attack can break the system with non-negligible probability.
With this approach, it is possible to adapt an existing scheme to technological advancements like increased computing power by increasing the security parameter. A detailed introduction to the concept of asymptotic security and the security parameter can be found in \cite{katz_2020_IntroductionModern}. The most essential aspects for our purposes are given in the following.

In the following, we focus on symmetric authentication schemes. That is, we assume that the communicating parties share common secret information, which is the \emph{secret key}. One of the most prominent symmetric authentication schemes is a \ac{mac}.
To ensure that the received message identifies its sender and has not been tampered with, a tag is added to the message using the shared secret key. In a nutshell, a \ac{mac} is a protocol consisting of the following three main steps: key generation, tagging, and verification.

\begin{definition}(MACS~\cite[Definition 4.1]{katz_2020_IntroductionModern})\label{def:MAC_scheme}
    Given a security parameter $n$, a \acf{mac} is an 
    authentication protocol given by a triple $(\keygen, \mathsf{MAC}, \verify)$ of probabilistic polynomial-time algorithms defined as follows:
    \begin{enumerate}
        \item The algorithm $\keygen$ randomly generates a key $k \in \set{0, 1}^N$ with $N \geq n$. 
        \item The authentication algorithm is a function $\mack$ with respect to a given $k \in \set{0,1}^N$ mapping a message $m \in \set{0,1}^\star$ of arbitrary length $\star$ to a \emph{tag} $t \in \set{0,1}^\star$, i.e., $t = \mack(m)$.
        \item Taking the inputs $k, m$ and $t$, the algorithm $\verify$ outputs a bit $b \in \set{0,1}$, with $b=1$ meaning valid and $b=0$ meaning invalid.
        \item For any positive integer $n$, any key $k \in \set{0,1}^N$ and any message $m\in \set{0,1}^\star$ the verification algorithm must satisfy $\verify_{k}(m, \mack(m)) = 1$.
    \end{enumerate}
\end{definition}

For a \ac{mac} to be secure, we want an attack 
to have only a ``small probability of success''. This means,
any attack's success probability should grow slower than $p(n)^{-1}$, where $p$ is a polynomial in the security parameter $n$. We refer to such a function as \emph{negligible} defined as follows.
\begin{definition}(Negligible Function~\cite[Definition 3.5]{katz_2020_IntroductionModern})\label{def:negligible}
    A function $f: \mathbb{N} \rightarrow \mathbb{R}_{\geq 0}$ is said to be \emph{negligible} if for every polynomial $p(n)$ there exists a positive integer $n^{*}$ such that for all integers $n > n^{*}$ it holds that $f(n)~<~p(n)^{-1}$. We denote $f \in \negl(n)$.
 \end{definition}

Mentioning success probabilities of attacks raises the question of what it means for a \ac{mac} to be broken. We assume that an adversary has access to an oracle version of the \ac{mac}, allowing the adversary to request a \ac{mac} tag for any message of their choice. Eventually, the \ac{mac} scheme is considered to be broken, if the adversary is able to forge a valid signature on any new message in polynomial time, not previously chosen and signed. We formally define this attack scenario as a \emph{security game} analogously to~\cite{katz_2020_IntroductionModern}.

\begin{definition}(Message Authentication Forgery Game $\macforge$)\label{def:MAC_security-game}
    Given a message authentication code system $\Pi = (\keygen, \mac, \verify)$, an adversary $\adversary$ with oracle access to the authentication algorithm and a positive integer $n$. The \emph{message authentication forgery game}, $\macforge_{\adversary, \Pi}(n)$, is defined as follows.
    \begin{enumerate}
        \item Choose a random key $k \in \set{0,1}^N$ with $N \geq n$. 
        \item The adversary $\adversary$ repeatedly chooses a message $m$ from the message space, and obtains the tag $t = \mac_k(m)$. Let $\mathcal{Q}$ denote the set of messages chosen by $\adversary$.
        \item The adversary $\adversary$ wins the security game, i.e., $\macforge_{\adversary, \Pi}(n)=1$, if and only if $\adversary$ is able to generate a message-tag pair $(m',t')$, where $m'~\notin~\mathcal{Q}$, that verifies correctly under $\verify_k$. That is, $\verify_k(m', t') = 1$.
    \end{enumerate}
\end{definition}

For a \ac{mac} scheme to be considered secure, no polynomial-time attack should succeed in the security game with non-negligible probability. 
This yields the notion of \ac{eufcma} security defined as follows.
\begin{definition}(EUF-CMA security~\cite[Definition 4.2]{katz_2020_IntroductionModern})\label{def:EUF-CMA}
    Let $\Pi = (\keygen, \mac, \verify)$ be a \ac{mac}. We say $\Pi$ is \acf{eufcma} if for all adversaries $\adversary$ capable of a probabilistic polynomial-time attack, we have
    \begin{align*}
        \prob{\macforge_{\adversary, \Pi}(n) = 1} \in \negl(n).
    \end{align*}
\end{definition}
\ac{eufcma} security is the standard security notion for digital signature schemes and authentication schemes required by the \ac{nist} \cite{nationalinstituteofstandardsandtechnologyus_2016_SubmissionRequirements}.

\subsection{Physically Unclonable Functions}\label{subsec:PUFs}
With the goal of ensuring hardware security, i.e., to protect a secret key in vulnerable hardware, \acp{puf} provide a viable way to securely regenerate keys (see, for instance, \cite{armknecht_2010_MemoryLeakageResilient,skoric_2005_RobustKey,lim_2005_ExtractingSecret}). The concept of a \ac{puf} was first introduced by Pappu in \cite{pappu_2002_PhysicalOneWay} using the term Physical One-Way Function, and shortly thereafter in~\cite{gassend_2002_SiliconPhysical} with the term Physical Random Function. 
In a nutshell, a \ac{puf} is a probabilistic function derived from the behavior of a complex physical object. Upon excitation by possible stimuli (called \emph{challenges}), the function outputs the respective responses. In the last few years, \acp{puf} have been used in more advanced cryptographic protocols (e.g., authentication schemes), emerging as a novel cryptographic primitive \cite{katzenbeisser_2011_RecyclablePUFs,vandijk_2012_PhysicalUnclonable}. In these cryptographic schemes, only the party currently holding the physical object can query for responses to chosen challenges. Due to the complexity of the \ac{puf}'s challenge--response behavior, it is not feasible to use any purely numerical method to predict a response to a randomly generated challenge without holding the \ac{puf} at this exact moment of time.

In recent years, various notions of \acp{puf} have been introduced, each capturing distinct security aspects. In this work, we adopt the framework presented in \cite{armknecht_2011_FormalizationSecurity}, which we later translate to chemical functions in \Cref{sec:chemical_functions}. There, \acfp{pf} are introduced, and the unclonability of the \ac{pf} is defined as one possible security property.
We recap briefly the framework of \acp{pf} presented in \cite{armknecht_2011_FormalizationSecurity}.
A new instantiation of a \ac{pf} can be obtained by the generation process $\gen$ that, upon a generation parameter $\alpha_{\mathrm{cr}}$ chosen by the \emph{manufacturer}, outputs a physical component $\physcomp$. 
A \ac{pf} strongly depends on the physical properties defined by $\physcomp$. These properties are also subject to uncontrollable noise, wherefore a \ac{pf} might produce two distinct responses upon the same challenge. To cope with this issue, an extraction algorithm $\extract$ is applied to make sure that different responses to the same challenge that differ only slightly are mapped to the same output. 

\subsection{Cryptographic Schemes Based on Chemical Functions}\label{subsec:DNA-based_schemes}
While \acp{puf} are already widely used in classical cryptographic schemes (as discussed in \Cref{subsec:PUFs}), the use of chemical functions only recently gained more attention in cryptographic applications (see, for e.g., \cite{clelland_1999_HidingMessages,vippathalla_2023_SecureStorage,leier_2000_CryptographyDNA,grass_2020_GenomicEncryption,li_2022_GeneticPhysical,volf_2023_CryptographyDNA,luescher_2024_ChemicalUnclonablea}). 
The analogy to \acp{puf} was shown in \cite{luescher_2024_ChemicalUnclonablea}, where the authors introduced their scheme as a \ac{cuf}.
For practical reasons, we will focus more specifically on the schemes presented in \cite{luescher_2024_ChemicalUnclonablea} and \cite{volf_2023_CryptographyDNA}, where \cite{luescher_2024_ChemicalUnclonablea} will serve as a running example throughout the course of this paper. Therefore, in this subsection, we will quickly recap the main points of the two schemes.

\subsubsection{Operable Random DNA \cite{luescher_2024_ChemicalUnclonablea}}\hfill\\
The concept of \ac{ordna} uses randomly manufactured items that are able to process a physical stimulus into an output that is unique to the respective input. Through the random features, it is impossible to reverse-engineer the input from the output, in analogy to a mathematical one-way function such as used in cryptographic hashes. The authors~\cite{luescher_2024_ChemicalUnclonablea} use up to $10^{10}$ randomly generated \ac{dna}-sequences. The specific instance of \ac{ordna} forms a substrate for a chemical computing unit, which is able to transform an input into an output. The input corresponds to a set of \ac{pcr} primers, which amplify a very small subset out of the billions of random sequences originally present in the \ac{ordna}. The amplified sequences are a specific fingerprint to the input (and the specific instance of random sequences) and can be identified using sequencing. The process is designed in such a way that the input cannot be read back from the output, and the random pool can be operated on but can not be copied by \ac{pcr}. We will explain the details of this scheme during the course of this paper in the examples.

\subsubsection{Genomic Sequence Encryption \cite{volf_2023_CryptographyDNA}}\hfill\\
\acf{gse} is a \ac{dna} sequencing-based encryption scheme that relies on the hardness of detecting rare point variations scattered across a large genome without prior knowledge of the modified locations. 
Thus, the secret key consists of a list of genomic coordinates in which edits can be introduced. Information is encoded using single-nucleotide edits across various positions of the genome. That is, information is encoded in binary where $0$ represents an unedited reference base and $1$ for a specific base edit. To edit the bases the authors use \acp{abe} (for converting $\A/\T$ to $\G/\C$) and \ac{cbe} (for converting $\G/\C$ to $\A/\T$). At the receiver's end, the decryption foresees targeted sequencing at the known corrupted coordinates, which can be efficiently performed at high coverage. From an adversarial point of view (i.e., without knowledge of the mutated genome coordinates), all genome bases must be scanned and analyzed for mutations, which are prone to sequencing errors and expensive to detect at deep coverage.

Although \ac{gse} was introduced as an encryption scheme, it can also be used for authentication. The authentication consists of determining the $0$/$1$ sequence in a given ordered subset of any set of locations as elaborated in \Cref{sec:genomic-sequence-encryption}.

\section{Framework for Chemical Functions}\label{sec:chemical_functions}
In this section, we introduce a generalized framework analogous to PUFs for analyzing authentication systems based on chemical substances, which we refer to as \acfp{cf}.
As \acp{cf} have a high degree of similarity to \acp{puf} and \acp{pf}, we draw parallels to the field and borrow some of the terminology, especially from \cite{armknecht_2011_FormalizationSecurity,ruhrmair_2009_FoundationsPhysical,ruhrmair_2010_ModelingAttacks}.
First, we formally define a \ac{cf} before introducing the surrounding infrastructure which we refer to as \acf{cfi}.
\begin{figure}[htbp]
    \centering
    \resizebox{\columnwidth}{!}{\begin{tikzpicture}[node distance=0cm, transform shape]

\definecolor{noisecolor}{HTML}{5c0404} 

\tikzset{
    converter/.style={
        rectangle, 
        draw, 
        fill=bleudefrance!20!green!5, 
        align=center, 
        inner ysep=10pt,
        minimum width=2.5cm,
        minimum height=2.2cm,
        execute at begin node=\setlength{\baselineskip}{15pt}
    },
    chemistry/.style={
        rectangle, 
        draw=green!30!black, 
        fill=green!15!white, 
        dashed, 
        align=center, 
        inner ysep=10pt,
        minimum width=3cm,
        minimum height=2.2cm,
        execute at begin node=\setlength{\baselineskip}{15pt}
    },
    noise/.style={
        noisecolor
    },
    arrow_style/.style={
        ->, 
        >=Stealth
    },
    system_box/.style={
        rectangle, 
        draw, 
        thick, 
        inner sep=10pt
    }
}

\node (chemical_reaction) [chemistry] {Chemical \\ Substance \\ with profile $\profile$};
\node (converter) [converter, left=1cm of chemical_reaction.west] {Chemical \\ Expression};
\node (measure) [converter, right=1cm of chemical_reaction.east] {Measurement};

\node (challenge) [left=0.9cm of converter] {$x$};
\node (response) [right=0.9cm of measure] {$y$};

\node (noise1) [above=1cm of converter, noise] {Noise};
\node (noise2) [above=1cm of chemical_reaction, noise] {Noise};
\node (noise3) [above=1cm of measure, noise] {Noise};

\begin{scope}[on background layer]
    
    \node (cf) [system_box, fill=bleudefrance!20, inner ysep=20pt, inner xsep=7pt,
                         fit=(chemical_reaction) (converter) (measure),
                         label={[anchor=south east]south east: $\cf{\profile}$: \acl{cf} with profile  $\profile$ }] {};

\end{scope}

\draw [arrow_style] (challenge) -- (converter);
\draw [arrow_style] (measure) -- (response);

\draw [arrow_style] (converter.east) -- node[above] {$\tilde{x}$} (chemical_reaction.west);
\draw [arrow_style] (chemical_reaction.east) -- node[above] {$\tilde{y}$} (measure.west);

\draw [noise, arrow_style] (noise1) -- (converter);
\draw [noise, arrow_style] (noise2) -- (chemical_reaction);
\draw [noise, arrow_style] (noise3) -- (measure);

\end{tikzpicture}}
    \caption{Detailed view of the \acf{cf} $\cf{\profile}$. The digital representation of the challenge $x$ is converted to chemical reactants $\tilde{x}$. These reactants are applied to the chemical substance with profile $\profile$. After the reaction takes place, the output $\tilde{y}$ is measured and converted to a digital representation $y$. Note that all three steps are noisy.}
    \label{fig:chemical_function_figure}
\end{figure}
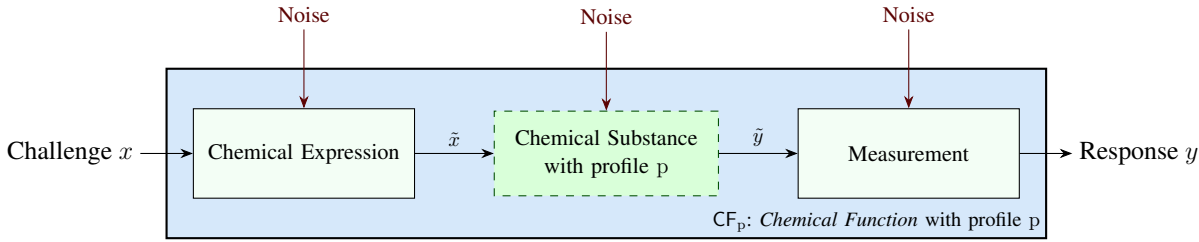
 
As shown in \Cref{fig:chemical_function_figure}, the central component of the \ac{cf} is a chemical substance. The behavior of a chemical substance is characterized by its exact chemical composition, which we refer to as its \emph{profile}, denoted by $\profile$. The profile $\profile$ fully characterizes the chemical composition of the substance and therefore the behavior of the \ac{cf}, and can be either measurable or not.
It is stimulated by an input signal $\tilde{x}$ to obtain a corresponding output signal $\tilde{y}$. This process happens in the chemical domain.
The input signal $\tilde{x}$ can be described by a digital representation $x$ which we refer to as the \emph{challenge}. The output signal $\tilde{y}$ can be measured and converted to a digital representation $y$, which we call the \emph{response}.
With these ideas in mind, we now begin by formally defining the chemical function.
\begin{definition}[Chemical Function]
    A \emph{chemical function} $\cf{}$ is a noisy mapping
    $ \cf{}^{\alpha_{E}} : \setchal \to \setresp $,
    where $\setchal$ is the set of all valid challenges, $\setresp$ is the set of all possible responses, and $\alpha_{E}$ is a fixed evaluation parameter. 
\end{definition}
For the sake of brevity, we omit mentioning $\alpha_{E}$ and denote a chemical function simply as $\cf{}$. The specific instantiation of a chemical function is denoted by $\cf{\profile}$. The profile $\profile$ fully characterizes the chemical composition of the substance and therefore the behavior of the \ac{cf}.

\begin{example}\label{ex:chemical_functions_simple}
    An example of a chemical function is a test strip that changes color based on the pH level of a solution. The chemical substance in the strip has a specific profile $\profile$ that determines how it reacts to different pH levels. The challenge $x$ is the pH level of the solution applied to the strip, and $\Tilde{x}$ the actual substance applied to the test strip.
    The response $\Tilde{y}$ is the color change observed on the strip and $y$ the measured color or wavelength histogram. The mapping from pH level to color change is noisy due to variations in lighting conditions, observer perception, and manufacturing inconsistencies in the test strips. Thus, this system can be modeled as a chemical function $\cf{\profile}$ that maps challenges (pH levels) to responses (colors) with some inherent noise.
    
    Another example of a chemical function is a mixture of fluorescence dyes that emit different colors when excited by light of a specific wavelength \cite{woidasky_2020_InorganicFluorescent}. The profile $\profile$ would describe the concentrations and types of dyes in the mixture. The challenge $x$ could be the wavelength of the excitation light, and the response $y$ would then be the resulting fluorescence spectrum. In this example, only one challenge would be available if the dyes had only a single excitation wavelength. Again, due to factors like dye interactions and environmental conditions, the response would be noisy.
    Beyond these examples, many digital barcoding systems can be modeled as chemical functions~\cite{yang_2024_DigitalBarcodes}.
\end{example}

Aside from the chemical function itself, an infrastructure is required. We refer to this entire infrastructure as the \acf{cfi}, denoted $\cfi$, which is shown in \Cref{fig:chemical_function_infrastructure_figure}.
The chemical function is manufactured through a publicly known generation process $\gen$, which takes as input the security parameter $n$, and is subject to some parameters $\alpha_G$ and innate synthesis variability. 
In contrast to physical functions, chemical functions will be generated in larger quantities such that they can be divided into smaller parts and distributed across space. In the special case of \ac{dna}-based chemical functions, even amplification of the chemical substance can be possible. 

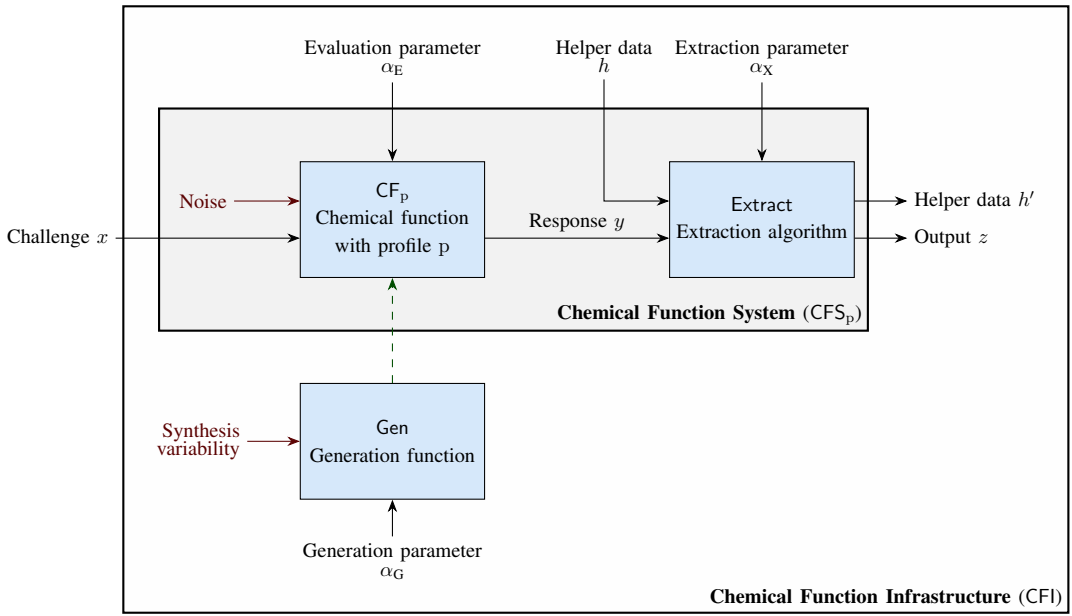
\begin{figure*}[htbp]
    \centering
    \begin{tikzpicture}[node distance=0cm, transform shape, scale=0.7]

\definecolor{noisecolor}{HTML}{5c0404} 

\tikzset{
    function/.style={
        rectangle, 
        draw, 
        fill=bleudefrance!20, 
        align=center, 
        inner ysep=10pt,
        minimum width=3.5cm,
        minimum height=2.2cm,
        execute at begin node=\setlength{\baselineskip}{15pt}
    },
    component/.style={
        rectangle, 
        draw=green!30!black, 
        fill=green!15!white, 
        dashed, 
        align=center, 
        inner ysep=10pt,
        minimum width=3.5cm,
        minimum height=1.2cm,
        execute at begin node=\setlength{\baselineskip}{15pt}
    },
    arrow_style/.style={
        ->, 
        >=Stealth
    },
    param_arrow/.style={
        ->, 
        >=Stealth, 
        dashed, 
        draw=green!30!black
    },
    label_style/.style={ 
        execute at begin node=\setlength{\baselineskip}{15pt}
    },
    system_box/.style={
        rectangle, 
        draw, 
        thick, 
        inner sep=7pt
    }
}

\node (chemical_function) [function] {$\cf{\profile}$ \\ Chemical function \\ with profile $\profile$};
\node (gen_function) [function, below=2cm of chemical_function.south, xshift=0cm] {$\gen$ \\ Generation function};
\node (extract_func) [function, right=3.5cm of chemical_function] {$\extract$ \\ Extraction algorithm};

\node (gen_param) [align=center, execute at begin node={\setlength{\baselineskip}{10pt}}, below=0.7cm of gen_function.south] {Generation parameter \\ $\alpha_{\text{G}}$};
\node (production_variability) [noisecolor, align=right, execute at begin node={\setlength{\baselineskip}{10pt}}, left=1cm of gen_function.west] {Synthesis \\ variability};

\node (eval_param) [align=center, execute at begin node={\setlength{\baselineskip}{10pt}},above=1.5cm of chemical_function.north] {Evaluation parameter \\ $\alpha_{\text{E}}$};
\node (noise) [noisecolor, left=1.25cm of chemical_function.west, yshift=\baselineskip] {Noise};
\node (challenge_x) [left=3.5cm of chemical_function.west, yshift=-\baselineskip] {Challenge $x$};

\node (extract_param) [align=center, execute at begin node={\setlength{\baselineskip}{10pt}},above=1.5cm of extract_func.north] {Extraction parameter \\ $\alpha_{\text{X}}$};
\node (helper_data_h) [align=center, execute at begin node={\setlength{\baselineskip}{10pt}}, anchor=north] at ($(extract_param.north)+(-3cm,0)$) {Helper data \\ $h$};
\node (helper_data_h2) [right=1cm of extract_func.east, yshift=\baselineskip] {Helper data $h'$};
\node (output_z) [right=1cm of extract_func.east, yshift=-\baselineskip] {Output $z$};

\begin{scope}[on background layer]
    
    \node (pfi_box) [system_box, fill=gray!1,
                         fit=(gen_function) (gen_param) (eval_param) (extract_param) (helper_data_h) (production_variability) (helper_data_h2) (output_z),
                         label={[anchor=south east]south east:\textbf{Chemical Function Infrastructure }($\cfi$)}] {};
    \node (cfs_box) [system_box, fill=gray!10, inner ysep=20pt, inner xsep=5pt,
                     fit=(extract_func) (noise) (chemical_function) (chemical_function),
                     label={[anchor=south east]south east:\textbf{Chemical Function System }($\cfs{\profile}$)}] {};

\end{scope}

\draw [arrow_style] (gen_param) -- (gen_function);
\draw [draw=noisecolor, arrow_style] (production_variability.east) -- (gen_function.west);
\draw [param_arrow] (gen_function.north) -- ( chemical_function.south);

\draw [arrow_style] (challenge_x) -- (challenge_x.east -| chemical_function.west);
\draw [noisecolor, arrow_style] (noise.east) -- (noise.east -| chemical_function.west);
\draw [arrow_style] (eval_param) -- (chemical_function);
\draw [arrow_style] ([yshift=-\baselineskip] chemical_function.east) -- node[above] {Response $y$} ([yshift=-\baselineskip] extract_func.west);

\draw [arrow_style] (extract_param) -- (extract_func);
\draw [arrow_style] (helper_data_h) |- ([yshift=\baselineskip] extract_func.west);
\draw [arrow_style] ([yshift=\baselineskip] extract_func.east) -- (helper_data_h2);
\draw [arrow_style] ([yshift=-\baselineskip] extract_func.east) -- (output_z);

\end{tikzpicture}
    \caption{Architecture of the \acf{cfi}, illustrating the three main layers: \acl{cfi} $\cfi$, \acl{cfs} $\cfs{}$, and \acl{cf} $\cf{}$. Each layer contains specific processes and parameters contributing to the overall authentication mechanism.}
    \label{fig:chemical_function_infrastructure_figure}
\end{figure*}

Due to ambient noise in the system, the response generated by the \ac{cf} is prone to errors and thus non-deterministic. To tackle this, \acp{cf} are usually coupled with an appropriate extraction algorithm, $\extract$, which processes the response $y$ and maps it to an output $z$ similar to \cite{dodis_2004_FuzzyExtractors}. 
The extraction algorithm relies on side information, which we refer to as the \emph{helper data}, indicating whether a challenge has been observed before. 

\begin{definition}[Extraction Algorithm]\label{def:extraction_algorithm}
    An \emph{extraction algorithm} is a mapping
    \begin{align*}
        \extract^{\para{X}} : \setresp \times \sethelper \to \setext \times \sethelper,
    \end{align*}
    where $\setresp$ is the set of responses, $\sethelper$ is the set of possible helper data, $\setext$ is the set of all possible extracted responses, and ${\para{X}}$ denotes the extraction parameters. 
    When a challenge $x \in \setchal$ is requested for the first time, the extraction is done in \emph{setup} mode with empty helper data $ \emptyset$, and the algorithm outputs the reconstructed response $z \in \setext$ along with helper data $h \in \sethelper$. Later, if the same challenge $x$ is requested again, the extraction algorithm uses the previously generated helper data $h$ to reconstruct $z$. The utilization of helper data in this manner enables a more robust mapping from the challenge to the output of the chemical function, as defined in \Cref{sec:properties}.
\end{definition}

Hence, the extraction algorithm takes as input the noisy response $y$ and some helper data $h$ and outputs the extracted response $z$ along with some helper data $h'$. If the input helper data is not empty, the algorithm outputs the same helper data again.
In a coding theoretic sense, a perfect extraction algorithm can be viewed as a mapping from the response space $\setresp$ to a set of Voronoi regions in $\setext$ defined by the helper data. The extraction algorithm is designed to map noisy responses that fall within the same Voronoi region to the same extracted response $z$. 

As the extraction algorithm is handling the noise in the response, it needs to be adapted to the noise level of the chemical function. For example, if different devices are used for measurement, the noise characteristics may differ significantly and the extraction algorithm needs to be adjusted accordingly. This helps to maintain a certain robustness of the system as defined in \Cref{sec:properties} but comes at the cost of reducing the number of possible extracted responses as the Voronoi regions need to be larger to accommodate the increased noise.

It is convenient for the rest of the paper to view the chemical function and the extraction algorithm as a single system. We refer to this combined system as the \ac{cfs}. Furthermore, for a chemical function $\cf{\profile}$, we denote the corresponding chemical function system as $\cfs{\profile}$.

\begin{definition}[Chemical Function System]
    A \acl{cfs} $\cfs{\profile}$ is a probabilistic mapping obtained by concatenating a \ac{cf} and an extraction algorithm, defined as
    \begin{align*}
        \cfs{\profile} : \setchal \times \sethelper  &\to \setext \times \sethelper \\
        (x, h) &\mapsto (z, h') 
    \end{align*}
    where $\setchal$ is the set of challenges, $\sethelper$ is the set of possible helper data and $\setext$ is the set of all possible outputs.
\end{definition}
Due to the embedded extraction algorithm, the \ac{cfs} can be in \emph{setup} or \emph{reconstruction} mode analogous to \Cref{def:extraction_algorithm}. The function also comprises an evaluation parameter $\alpha_E$ and an extraction parameter $\para{X}$ that are fixed for a given system and are, once again, omitted.

Note that even though the extraction algorithm helps mitigate noise in the response, the output $z$ is still prone to errors due to the inherent stochasticity of chemical processes and measurement. We can now apply these definitions to the example of fluorescence dyes introduced in Example \ref{ex:chemical_functions_simple}.

\begin{example}
    \label{ex:dye_extract}
    Consider a chemical function based on a mixture of up to $7$ fluorescence dyes that emit different colors when excited by light of a specific wavelength.  
    The response $y$ represents the information of the presence or absence of a specific color. This gives us a $7$-bit vector. Through some noise in the measurement process, 1 color might be missed or an additional color might be detected. 
    The extraction algorithm then may be a $(7,4,3)$ Hamming code that maps the $7$-bit vector to a $4$-bit extracted response $z$. The helper data can be empty in this case. It will be used if the Voronoi region cannot be defined through a perfect code. Using the Hamming code in the extraction algorithm, if one color is missed or an additional color is detected, the extraction algorithm can still recover the same $4$-bit response $z$ for subsequent experiments.
\end{example}

We now consider the \ac{ordna} scheme, which requires a more sophisticated extraction algorithm.
\begin{example}
    Another example of a \ac{cf} is the \ac{ordna} scheme presented in \cite{luescher_2024_ChemicalUnclonablea}, which we briefly introduce here in an abstracted form. In this scheme, the chemical substance is a pool of random \ac{dna} sequences generated using the $\gen$ process. The profile $\profile$ of the chemical function $\cf{\profile}$ is the specific set of random \ac{dna} sequences present in the pool. The structure of the constant and random parts of the \ac{dna} molecules is illustrated in \Cref{fig:ordna_design}.
    \begin{figure}[h]
        \centering
        \resizebox{\columnwidth}{!}{\begin{tikzpicture}[node distance=0cm, transform shape, scale=0.7]
    \tikzset{
        constant/.style={draw, fill=black, minimum height=0.5cm, outer sep=0, text=white},
        random/.style={draw, fill=lightgray, minimum height=0.5cm, outer sep=0, text=black},
        label/.style={above, midway, yshift=0.2cm},
        brace/.style={decorate, decoration={brace, amplitude=5pt}}
    }

    \node[constant, minimum width=2.0cm] (handle1) {};
    \node[random,   minimum width=2.0cm, right=of handle1] (input1) {};
    \node[constant, minimum width=2.0cm, right=of input1] (adapter1) {};
    \node[random,   minimum width=2.0cm, right=of adapter1] (output) {};
    \node[constant, minimum width=2.0cm, right=of output] (adapter2) {};
    \node[random,   minimum width=2.0cm, right=of adapter2] (input2) {};
    \node[constant, minimum width=2.0cm, right=of input2] (handle2) {};

    \def \YSHIFT{0.1cm}
    \def \XSHIFT{0.05cm}
    \draw[brace] ([yshift=\YSHIFT, xshift=\XSHIFT]handle1.north west) -- ([yshift=\YSHIFT, xshift=-\XSHIFT]handle1.north east) node[label, align=center] {Handle I\\(20 nt)};
    \draw[brace] ([yshift=\YSHIFT, xshift=\XSHIFT]input1.north west) -- ([yshift=\YSHIFT, xshift=-\XSHIFT]input1.north east) node[label, align=center] {Input I\\(9 nt)};
    \draw[brace] ([yshift=\YSHIFT, xshift=\XSHIFT]adapter1.north west) -- ([yshift=\YSHIFT, xshift=-\XSHIFT]adapter1.north east) node[label, align=center] {Adapter I\\(20 nt)};
    \draw[brace] ([yshift=\YSHIFT, xshift=\XSHIFT]output.north west) -- ([yshift=\YSHIFT, xshift=-\XSHIFT]output.north east) node[label, align=center] {Output\\(21 nt)};
    \draw[brace] ([yshift=\YSHIFT, xshift=\XSHIFT]adapter2.north west) -- ([yshift=\YSHIFT, xshift=-\XSHIFT]adapter2.north east) node[label, align=center] {Adapter II\\(21 nt)};
    \draw[brace] ([yshift=\YSHIFT, xshift=\XSHIFT]input2.north west) -- ([yshift=\YSHIFT, xshift=-\XSHIFT]input2.north east) node[label, align=center] {Input II\\(10 nt)};
    \draw[brace] ([yshift=\YSHIFT]handle2.north west) -- ([yshift=\YSHIFT]handle2.north east) node[label, align=center] {Handle II\\(20 nt)};

    \draw[red, thick, |-|] ([yshift=-0.3cm, xshift=-0.7cm]input1.south west) -- ([yshift=-0.3cm, xshift=-0.1cm]input1.south east);

    \draw[red, thick, |-|, yshift=-0.6cm] ([yshift=-0.3cm, xshift=0.1cm]input2.south west) -- ([yshift=-0.3cm, xshift=0.7cm]input2.south east);

    \draw[green!50!black, dash pattern= on 3pt off 5pt, dash phase=4pt, thick] ([xshift=0.27cm, yshift=0.1cm]handle1.north) -- ++(0,-1.2cm);
    \draw[green, dash pattern= on 3pt off 5pt, thick] ([xshift=0.27cm, yshift=0.1cm]handle1.north) -- ++(0,-1.2cm);
    \draw[green!50!black, dash pattern= on 3pt off 5pt, dash phase=4pt, thick] ([xshift=-0.27cm, yshift=0.1cm]handle2.north) -- ++(0,-1.2cm);
    \draw[green, dash pattern= on 3pt off 5pt, thick] ([xshift=-0.27cm, yshift=0.1cm]handle2.north) -- ++(0,-1.2cm);

    \begin{scope}[xshift=-0.75cm, yshift=-1.5cm]
        \node[constant, minimum height=0.25cm, minimum width=0.5cm] (constant) {};
        \node[right, right=of constant] {Constant/determined part};
        \node[random, below=0.2cm of constant, minimum height=0.25cm, minimum width=0.5cm] (random) {};
        \node[right, right=of random] {Randomly synthesized part};
        \node[below=0.2cm of random, minimum height=0.25cm, minimum width=0.5cm] (cut) {};
        \draw[green!50!black, dash pattern= on 3pt off 5pt, dash phase=4pt, thick] (cut.west) -- (cut.east);
        \draw[green, dash pattern= on 3pt off 5pt, thick] (cut.west) -- (cut.east);
        \node[right, right=of cut] {Blunting site (preventing amplification)};
        \node[below=0.2cm of cut, minimum height=0.25cm, minimum width=0.5cm] (selection) {};
        \draw[red, thick, |-|] (selection.west) -- (selection.east);
        \node[right, right=of selection] {Selection PCR primer binding site};
    \end{scope}
\end{tikzpicture}}
        \caption{The structure of an \ac{ordna} molecule. The outer handles can be used to amplify the sequences. The outer handles are blunted to prevent further amplification. The random inputs correspond to the challenges, as the selection \ac{pcr} binds with these to perform the amplification process on these strands only. The random output part of multiple reads corresponds to the responses of the \ac{cf}. Image replicated from \cite{luescher_2024_ChemicalUnclonablea}.}
        \label{fig:ordna_design}
    \end{figure}
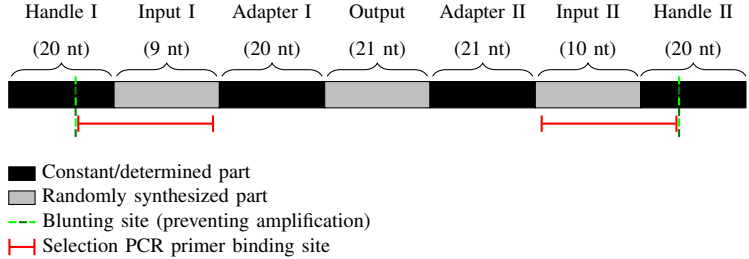
    The challenge $x$ to the chemical function consists of a set of \ac{pcr} primers, which amplify a small subset of sequences from the pool defined by the random input part.
    When this challenge is applied, only the strands whose random input section match are amplified exponentially. Therefore, these strands will make up the vast majority of the pool. If the pool is sequenced, with high probability one of these strands will be read. 
    The response $y$ is the set of random output parts of the amplified sequences. Due to the stochastic nature of \ac{pcr} and sequencing, the response $y$ is noisy. To obtain a robust output, an appropriate extraction algorithm $\extract$ is applied to $y$, which outputs $z$ along with some helper data $h$.
\end{example}

\section{Security Properties for Chemical Functions}\label{sec:properties}
Now that we have defined the chemical functions framework, we can define properties that a chemical function should satisfy to be suitable for security applications. We adapt the properties defined for \acp{puf} in \cite{armknecht_2011_FormalizationSecurity,ruhrmair_2009_FoundationsPhysical,ruhrmair_2010_ModelingAttacks} to our setting. The three main properties are: 
\begin{inparaenum}
    \item[i)] robustness, 
    \item[ii)] unclonability, and 
    \item[iii)] unpredictability.
\end{inparaenum}  
In contrast to \cite{armknecht_2011_FormalizationSecurity,ruhrmair_2009_FoundationsPhysical}, we will define these properties in a finite domain as well as in an asymptotic manner.

\subsection{Robustness}
Similar to \acp{puf}, chemical processes are noisy, meaning that the same challenge may not yield the exact same raw response on every evaluation. Therefore, an extraction algorithm $\extract$, as in \Cref{def:extraction_algorithm}, is employed to distill a stable output from noisy measurements. 
The extraction algorithm must be tailored to the noise model of a specific \ac{cf}.
Informally, robustness is the probability that repeated evaluations of the same \ac{cf} on the same challenge agree after extraction. Formally, we follow \cite{armknecht_2011_FormalizationSecurity} as the general setting is applicable here.

\begin{definition}[Robustness]\label{def:robustness}
Let $\cfs{\profile}$ be a \acl{cfs} with profile $\profile$ and let $x \in \setchal$ be a challenge. The \emph{challenge robustness} of $\cfs{\profile}$ with respect to $x$ is defined as the probability
\begin{multline*}
    \rho (\cfs{\profile}, x) \\
    = \prob{ \cfs{\profile} (x,h) = (z,h) | \cfs{\profile} (x,\emptyset) = (z, h)}.
\end{multline*}
\end{definition}
 
A first (setup) evaluation on $x$ yields $(z,h)$; a later (reconstruction) evaluation on the same $x$ using $h$ should reproduce the same $z$. The quantity $\rho(\cfs{\profile}, x)$ is exactly this agreement probability, absorbing stochastic variability in the chemistry and any randomized steps in extraction.

\begin{example}
    Continuing the simplified example of fluorescence dyes from \Cref{ex:dye_extract}, let us assume that the measurement process has a chance of missing a color or falsely detecting an additional color with probability $p$. Thus, when measuring the same $\cfs{\profile}$ multiple times with the same challenge $x$, the chance that any $k$ of the $7$ bits in the response $y$ are flipped follows a binomial distribution with parameters $7$ and $p$. Using a $(7,4,3)$ Hamming code as extraction algorithm, up to $1$ bit error can be corrected. The challenge's robustness $\rho(\cfs{\profile}, x)$ is then given as the probability of having at most $1$ bit error in the response which can be calculated as 
    \begin{align*}
        \rho (\cfs{\profile}, x) 
        &= \sum_{k=0}^{1} \binom{7}{k} p^k (1-p)^{7-k}\\
        &= (1-p)^7 + 7p(1-p)^6.
    \end{align*}
\end{example}

If we have a \ac{cf} with multiple challenges, we also want to define robustness over a set of challenges. This can be done both for a single \ac{cfs} instance as well as for a whole \ac{cfi} when considering also all possible instantiations of the \ac{cfs}.

\begin{definition}[Robustness of a \ac{cfs}]\label{def:robustness_cfs}
    Let $X$ be a random variable with realizations $x \in \setchal$.
    The minimum and average robustness of a \acl{cfs} $\cfs{\profile}$ with respect to the set of challenges $\setchal$ are defined as
    \begin{align*}
        \rho_{\min}(\cfs{\profile}, \setchal) &= \min_{x \in \setchal} \left\{ \rho (\cfs{\profile}, x) \right\},\\
        \rho_{\text{avg}}(\cfs{\profile}, \setchal) 
        &= \sum_{x\in\setchal} \prob{X = x}\rho(\cfs{\profile},x).
    \end{align*}
\end{definition}

In this definition, $\rho_{\min}$ provides a worst-case guarantee over the chosen challenge set, while $\rho_{\text{avg}}$ summarizes a typical behavior under a specified challenge distribution. Which one is relevant depends on the specific \ac{cfi} and the purpose it is used for.

\begin{definition}[Robustness of a \ac{cfi}]\label{def:robustness_cfi}
    Consider a \acl{cfi} $\cfi$ and let $\cfs{\profile}$ be a \acl{cfs} instance of $\cfi$ with profile $\profile \in \setprofile$ using the same creation, evaluation and extraction parameters. With $\setprofile$ we denote the set of all valid profiles as defined by the manufacturer. Given a random variable $P$ with realizations $\profile \in \setprofile$, we define the minimum and average robustness of a \ac{cfi} as
    \begin{align*}
        \rho_{\min}(\cfi, \setchal) &= \min_{\profile \in \setprofile} \left\{  \rho_{\min}(\cfs{\profile}, \setchal) \right\},\\
        \rho_{\text{avg}}(\cfi, \setchal) &= \sum_{\profile \in \setprofile} \prob{ P = p } \rho_{\text{avg}}(\cfs{\profile}, \setchal).
    \end{align*}
\end{definition}
This lifts robustness from a single instance to an implementation family. The minimum over $\setprofile$ captures the least reliable valid profile, whereas the average reflects expected behavior across the manufacturer's admissible profiles. 

With these notions in place, we next instantiate them on the \ac{ordna} construction to quantify empirical reliability under realistic experimental noise and extraction parameters. We continue the example from \cite{luescher_2024_ChemicalUnclonablea} to illustrate the concept of robustness. For simplicity, we will only analyze the average robustness of the \ac{cfi}. 

\begin{example} \label{ex:robustness}
    The extraction algorithm presented in \cite{luescher_2024_ChemicalUnclonablea} proceeds as follows. First, a filtering step removes artifacts and potential contamination from the raw reads obtained from applying the challenge $x$. Second, $k$-mer extraction is performed with $k=8$ according to \cite{liu_2022_CMashFast} on the middle part of the molecule labeled as \emph{Output} in \Cref{fig:ordna_design}. Third, the MinHash algorithm from \cite{broder_1998_ResemblanceContainment, haveliwala_2000_ScalableTechniques} is applied $30$ times to reduce the dimensionality of the data and obtain $w_i \in \mathbb{F}_{2^8}^{255}$ for $i \in \{1, ..., 30\}$. 
    This preprocessing simplifies the noise model of the scheme to a discrete memoryless $q$-ary symmetric channel with $q = 2^8$ and some error probability $p_{\text{noise}}$.
    Finally, a fuzzy vault scheme is applied \cite{juels_2006_FuzzyVault}: The resulting data is added to a random codeword of a Reed-Solomon code of length $n=255$ and redundancy $r=223$ to obtain the helper data $h_i = c_i + w_i$. The code's minimum Hamming distance is $d=224$, allowing for the guaranteed correction of up to $t = \lfloor (d-1)/2 \rfloor = 111$ byte errors. This corresponds to a noise level of up to $43.5\%$ in the MinHash data that can be corrected. After error correction, the ratio of successful decodings of the $30$ runs is computed and then mapped to a binary value via a thresholding method. To check against other samples, the difference of the helper data and the MinHash results $w_i'$ of that other sample are decoded with the Reed-Solomon code and compared to the output of the first sample.

    For simplicity, we assume that a single fixed MinHash function is used and that the byte errors are independent and identically distributed. Furthermore, we assume that all experiments are done on the same device using the same protocol and that the byte errors are independent of the profile $\profile \in \setprofile$, the challenge $x \in \setchal$, and the pool size. While the actual noise model is more complex and probably not i.i.d., the binomial model is a sufficient approximation to determine the results for the robustness. Based on the published data and code \cite{luescher_2024_ChemicalUnclonablea}, the mean byte error rate after filtering and MinHash application is approximately $17\%$.

    Let $X$ be the random variable representing the number of byte errors in a $255$-byte string. We model $X$ using a binomial distribution with parameters $n=255$ and $p_{\text{noise}}=0.17$. The average robustness of the $\cfi$ is then given as the probability of a decoding success, i.e., the event of encountering at most $111$ byte errors, that is calculated as
    \begin{align*}
        \rho_{\text{avg}}(\cfi, \setchal) 
        &= \prob{X \leq 111} 
        \\
        &\approx 1 - 1.7 \times 10^{-15}.
    \end{align*}
    We note that, as indicated above, all profiles generated under \cite{luescher_2024_ChemicalUnclonablea} have a similar behavior. Namely, $\rho_{\min}(\cfs{\profile}, \setchal)$ is constant for all $\profile \in \setprofile$ and can be moved out of the sum.
    The set of challenges $\setchal$ considered in this example, inferring $p_{\text{noise}} = 0.17$, consists of the challenges that are analyzed in \cite{luescher_2024_ChemicalUnclonablea}. However, we can assume that the robustness will be similar for other challenges as well.
\end{example}

\subsection{Unclonability}\label{sec:unclonability}
For electronic \acp{puf}, the notion of unclonability is typically phrased as the infeasibility of fabricating a physical duplicate of the device. In the chemical setting, the relevant barrier instead concerns reconstructing by chemical or analytical means some profile that governs the \acp{cf} responses under admissible challenges.
Thus, unclonability requires that no efficient attack, even after interacting with an authentic \ac{cf}, can engineer a counterfeit \ac{cf} whose challenge--response behavior is (beyond negligible probability) indistinguishable from that of the genuine pool. 

Adhering to Kerckhoffs' principle, we assume that manufacturing details as well as the generation, evaluation, and extraction parameters $\alpha_{G}, \alpha_{E}, \alpha_{X}$ are public. An adversary's objective is therefore to recover (or emulate) the profile $\profile$ associated with the authentic chemical function $\cf{\profile}$, which plays an analogous role to a cryptographic key. It is realistic to assume that an adversary $\adversary$ can obtain one or more instances of the protected product and elicit a polynomial number of function evaluations to inform a cloning attempt.
At a high level, three generic attack methods are available:
\begin{enumerate}
    \item \textbf{Dilution:} In this approach, an adversary $\adversary$ obtains samples of the original product containing the authentic $\cf{\profile}$ and reuses them. 
    Due to the nature of a \ac{cfi} being divisible, $\adversary$ can reuse the original $\cf{\profile}$ in a diluted form to produce responses to new challenges. A potential countermeasure is to specify the expected concentration of the \ac{cf} in the genuine product, making significant dilution detectable.
    \item \textbf{Chemical Amplification:} In case the \ac{cf} is made of \ac{dna}, through techniques like \ac{pcr}, an attacker could amplify the \ac{cf} directly in a laboratory, exploiting \ac{dna}’s natural ability to replicate without exposing its full content. This enables the creation of counterfeit \acp{cf} via direct cloning of the original sequences. Existing schemes have put efforts in preventing this type of attack by chemical means~\cite{luescher_2024_ChemicalUnclonablea}.
    \item \textbf{Measure and Recreate:} An adversary analyzes $\cf{\profile}$ to obtain knowledge about the profile $\profile$ and then creates $\cf{\profile'}$ that provides the same challenge--response behavior. Mitigation includes designing \acp{cf} that are exceptionally difficult (or currently impractical) to analyze or that degrade under standard preparation protocols, hindering attempts to recreate them.
\end{enumerate}

In the case of \ac{dna}-based chemical functions, any means to amplify them must be prevented for the distributed \ac{cf}. This has to be done by chemical means. Therefore, we do not include this attack in the mathematical analysis of this work and assume that a suitable process is in place.
Instead, we focus on the third attack vector in this work.

The cost required for measurements and recreations can be determined analytically. 
Instead of providing a detailed cost analysis, we focus on defining unclonability in a cryptographic sense. Therefore, we bound the effort of the adversary and the effort in evaluating the \ac{cf} via the security parameter $n$ as in \Cref{subsec:authentication}. 
In contrast to \cite{ruhrmair_2009_FoundationsPhysical,armknecht_2011_FormalizationSecurity}, we analyze the chemical function in a finite domain as well as in an asymptotic regime, as it is relevant to see the unclonability of the scheme in an evolving context.
We define the adversarial model through a \emph{security game} which specifies the adversary's capabilities and the conditions for a successful \emph{break} in the spirit of \cite{katz_2020_IntroductionModern}.

\begin{definition}[Chosen Challenge Attack Cloning Game]\label{def:cloning_game}
    For a \acl{cfs} $\cfs{}$ and an adversary $\adversary$, the \emph{chosen challenge attack cloning security game} $\attack_{\adversary, \cfs{}}(n)$ is defined as follows.
    \begin{itemize}
        \item An authentic $\cfs{\profile}$ with profile $\profile \in \setprofile$ is generated via $\gen(n)$.
        \item The adversary $\adversary$ is given oracle access to $\cf{\profile}$ of $\cfs{\profile}$ and may adaptively query challenges $x_1, \ldots, x_q$ of their choice, receiving corresponding responses $y_1, \ldots, y_q$ and outputs $z_1, \ldots, z_q$. Let $\mathcal{Q}_{\adversary} = \{x_1, \ldots, x_q\}$ denote the set of queried challenges.
        \item The adversary produces a chemical function $\cf{\profile'}$ with profile $\profile' \in \setprofile$ such that the corresponding \acl{cfs} $\cfs{\profile'}$ utilizes the same extraction algorithm as $\cfs{\profile}$. The adversary succeeds, if for a randomly chosen fresh challenge $x \notin \mathcal{Q}_{\adversary}$ it holds that $\cfs{\profile'}(x,h)~=~\cfs{\profile}(x,h)$, where the helper data $h$ is obtained through $(z,h)~=~\cfs{\profile}(x, \emptyset)$. In this case, the outcome of the game is defined as $\attack_{\adversary, \cfs{}}(n) = 1$.
    \end{itemize}
\end{definition}

Note that the entire extraction algorithm is public. Therefore, an adversary $\adversary$ can also obtain the outputs $z_i$ and their corresponding helper data $h_i$ for each query by only having oracle access to $\cf{\profile}$.
This game captures the intuitive goal of an adversary: to create a counterfeit item that behaves identically to an authentic one for a new, previously untested challenge. The oracle access models an adversary who can purchase an authentic product and subject it to a polynomial number of tests to learn its behavior. Remember that we restrict the adversary to polynomial effort in $n$. This also applies to the third step of producing the counterfeit $\cf{\profile'}$.

As the adversary might be able to analyze the \ac{cf} with different methods, it is suitable to define another variant of the cloning game where the adversary is not restricted to querying challenges but can perform any evaluation of equivalent effort. We call this variant the \emph{open cloning game}, defined as follows.
\begin{definition}[Open Cloning Game]\label{def:open_cloning_game}
    For a \acl{cfs} $\cfs{}$ and an adversary $\adversary$, the \emph{open cloning security game} $\openattack_{\adversary, \cfs{}}(n)$ is defined as follows.
    \begin{itemize}
        \item An authentic $\cfs{\profile}$ with profile $\profile \in \setprofile$ is generated via $\gen(n)$.
        \item The adversary $\adversary$ can perform $q$ operations with the complexity of evaluating the \ac{cf} on $\cf{\profile}$ to obtain information.
        \item The adversary produces a chemical function $\cf{\profile'}$ with profile $\profile' \in \setprofile$ such that the corresponding \acl{cfs} $\cfs{\profile'}$ utilizes the same extraction algorithm as $\cfs{\profile}$. The adversary succeeds, if for a randomly chosen fresh challenge $x \notin \mathcal{Q}_{\adversary}$ it holds that $\cfs{\profile'}(x,h) = \cfs{\profile}(x,h)$, where the helper data is obtained through $(x,h)~=~\cfs{\profile}(x, \emptyset)$. In this case, the outcome of the game is defined as $\openattack_{\adversary, \cfs{}}(n) = 1$.
    \end{itemize}
\end{definition}

Now we can define unclonability in both an asymptotic and finite regime. To avoid unnecessary repetition, we only present the definitions for the chosen challenge attack cloning game of \Cref{def:cloning_game}. The definitions for the open cloning game of \Cref{def:open_cloning_game} are analogous.

\subsubsection{Finite Domain Unclonability}

We now define unclonability for the \ac{cfs} with profile $\profile$ in terms of an adversary's probability in succeeding in the cloning game after querying at most $q$ challenges. This captures the practical scenario where an adversary has limited interaction with the authentic system and must produce a convincing clone.

\begin{definition}[Finite Domain Unclonability]\label{def:finite_domain_unclonability}
    Let $\cfs{\profile}$ be a \acl{cfs} with profile $\profile \in \setprofile$ using the security parameter $n$.
    Let an adversary $\adversary$ play the cloning game of \Cref{def:cloning_game} on $\cfs{\profile}$, issuing at most $q$ (adaptive) challenge queries and thus obtaining $q$ responses.  
    We say that $\cfs{\profile}$ is $(\tau,q)$-\emph{unclonable under chosen challenge attack} if for all such adversaries it holds that
    \begin{align*}
        \prob{\attack_{\adversary, \cfs{\profile}}(n) = 1} \leq \tau.
    \end{align*}
\end{definition}
Note that the cloning game gives an adversary $\adversary$ the power of obtaining some challenge response pairs of $\cfs{\profile}$. The parameters $\tau$ and $q$ are related. For a fixed $\tau$, the manufacturer seeks the maximum number of queries $q$ in order to guarantee that the success probability stays below the threshold $\tau$. Or vice versa, for a fixed number of queries $q$, we can determine the minimum $\tau$ that can be achieved by a protocol used by the manufacturer.

\begin{example}
    We continue the example from \cite{luescher_2024_ChemicalUnclonablea} to illustrate finite-domain unclonability.
    A trivial extreme case is that an adversary can obtain responses for all challenges, i.e., $q = |\setchal|$. In this case, the adversary can fabricate a clone that reproduces the authentic \ac{cf} on every challenge. The probability of success for $\adversary$ is then bounded by the robustness of the \ac{cfi}, i.e., $\tau~\leq~\rho_{\min}(\cfi, \setchal)$. This holds for any \ac{cfi}; in particular, any $\cfi$ is $(\rho_{\min}(\cfi, \setchal), |\setchal|)$-unclonable.

    Now we look at the other extreme where the adversary does not have any knowledge about previous challenges (i.e. $q=0$).
    To bound $\tau$, we formalize that adversarial success reduces to a Reed–Solomon decoding event. As the vectors are obtained through a MinHash algorithm \cite{broder_1998_ResemblanceContainment}, the expected distance is modeled as a binomial random variable with success parameter equal to the (weighted) Jaccard similarity of the underlying $k$-mer multisets. For details on $k$-mer analysis as well as Jaccard similarity and its MinHash approximation, we refer the reader to \cite{compeau_2018_BioinformaticsAlgorithms, leskovec_2014_MiningMassive}. We upper-bound this similarity by exploiting permutation symmetry and collision counting over supports as elaborated in Appendix \ref{sec:orDNA-unclonability}.
    This leads to an upper bound on the expected number of byte errors $p_e$ in the MinHash output for a random pool of
    \begin{align*}
        p_e < \frac{1}{200} \sum_{i=0}^{200} i \, \frac{\binom{200}{i} \binom{4^8 - 200 - i}{\,200-i\,}}{\binom{4^8}{200}} < 0.0031.
    \end{align*}
        Using this bound we can upper-bound the probability $\tau$ of a successful cloning attack as the probability of having at most $t=111$ byte errors in the MinHash output, which is given by the tail of the binomial distribution with parameters $n=255$ and $p_e < 0.0031$:
    \begin{align*}
        \tau &= \prob{W \geq n-t | x \notin \mathcal{Q}_{\adversary}}
        < 10^{-280}.
    \end{align*}
    We conclude that the \ac{cfs} of \cite{luescher_2024_ChemicalUnclonablea} is approximately $(10^{-280}, 0)$-unclonable. Owing to the large amount of randomness in the \ac{ordna} pool, this extremely small success probability supports the intuition that synthesizing a random pool will almost surely not yield a functioning clone. 
    We want to point out that this analysis relies on several simplifications and assumptions, e.g., the uniformity of the $k$-mer distribution and the independence of MinHash outputs. Furthermore, the data is based on few experimental results rather than extensive statistical evidence. A more detailed analysis can start with the actual randomness in the pool of DNA like in \cite{meiser_2020_DNASynthesis} and evaluate all data processing steps. This would lead to a much more complex analysis with many parameters to consider. Therefore, we would recommend for a practical validation of a \ac{cfi} to perform empirical experiments on the error rates after the filtering and MinHash steps to obtain a noise model.
\end{example}

\subsubsection{Asymptotic Unclonability}
For defining unclonability in an asymptotic regime regarding $n$, we follow the concepts of \Cref{subsec:authentication} and \cite{katz_2020_IntroductionModern}.
We say a \ac{cfi} is \emph{asymptotically unclonable} if no adversary can win the cloning game $\attack_{\adversary, \cfs{}}$ with non-negligible probability for any valid instantiation $\cfs{}$ of $\cfi$. This provides a concrete guarantee that creating a functional forgery is computationally infeasible.

\begin{definition}[Asymptotic Unclonability]\label{def:asymptotic_unclonability}
    A \acl{cfs} $\cfs{}$ is \emph{asymptotically unclonable under chosen challenge attack} if for all adversaries $\adversary$ capable of polynomial-time attacks it holds that
    \begin{align*}
        \prob{\attack_{\adversary, \cfs{\profile}}(n) = 1} \in \negl(n).
    \end{align*}
\end{definition}

Assuming a \ac{cfi} that can be instantiated with any security parameter $n$, an adversary's success probability in winning the cloning game must be negligible in $n$. We continue the example from \cite{luescher_2024_ChemicalUnclonablea} to illustrate the concept of unclonability. 

\begin{example}\label{ex:ordna_asymptotic_unclonability}
The scheme in \cite{luescher_2024_ChemicalUnclonablea} is presented for a fixed length of challenges and output signals. Therefore, it is trivial to see that this finite system cannot be asymptotically unclonable. 
Without analyzing practical feasibility, we could, however, consider a multi-stage approach in the selection \ac{pcr} process. Assume that multiple random parts are added to the DNA molecules and multiple rounds of selection \ac{pcr} are performed. In each round, a different random part is used for the selection \ac{pcr}. An example of this idea is presented in \Cref{fig:2-stage_ordna_design} for a two-stage selection \ac{pcr}. In the first stage, a random part is selected and amplified. In the second stage, a different random part is used for further amplification. This process could be repeated multiple times. The final pool will then contain only strands that have the correct sequences in all random parts.

\begin{figure*}
    \centering
    \begin{tikzpicture}[node distance=0cm, transform shape, scale=0.6]
    \tikzset{
        constant/.style={draw, fill=black, minimum height=0.5cm, outer sep=0, text=white},
        random/.style={draw, fill=lightgray, minimum height=0.5cm, outer sep=0, text=black},
        label/.style={above, midway, yshift=0.2cm},
        brace/.style={decorate, decoration={brace, amplitude=5pt}}
    }

    \node[constant, minimum width=2.0cm] (handle1) {};
    \node[random,   minimum width=2.0cm, right=of handle1]  (input1) {};
    \node[constant, minimum width=2.0cm, right=of input1]   (adapter1) {};
    \node[random,   minimum width=2.0cm, right=of adapter1] (input3) {};
    \node[constant, minimum width=2.0cm, right=of input3]   (adapter3) {};
    \node[random,   minimum width=2.0cm, right=of adapter3] (output) {};
    \node[constant, minimum width=2.0cm, right=of output]   (adapter4) {};
    \node[random,   minimum width=2.0cm, right=of adapter4] (input4) {};
    \node[constant, minimum width=2.0cm, right=of input4]   (adapter2) {};
    \node[random,   minimum width=2.0cm, right=of adapter2] (input2) {};
    \node[constant, minimum width=2.0cm, right=of input2]   (handle2) {};

    \def \YSHIFT2{0.1cm}
    \def \XSHIFT2{0.05cm}
    \draw[brace] ([yshift=\YSHIFT2, xshift=\XSHIFT2]handle1.north west) -- ([yshift=\YSHIFT2, xshift=-\XSHIFT2]handle1.north east) node[label, align=center] {Handle I\\(20 nt)};
    \draw[brace] ([yshift=\YSHIFT2, xshift=\XSHIFT2]input1.north west) -- ([yshift=\YSHIFT2, xshift=-\XSHIFT2]input1.north east) node[label, align=center] {Input I\\(9 nt)};
    \draw[brace] ([yshift=\YSHIFT2, xshift=\XSHIFT2]adapter1.north west) -- ([yshift=\YSHIFT2, xshift=-\XSHIFT2]adapter1.north east) node[label, align=center] {Adapter I\\(20 nt)};
    \draw[brace] ([yshift=\YSHIFT2, xshift=\XSHIFT2]input3.north west) -- ([yshift=\YSHIFT2, xshift=-\XSHIFT2]input3.north east) node[label, align=center] {Input III\\(9 nt)};
    \draw[brace] ([yshift=\YSHIFT2, xshift=\XSHIFT2]adapter3.north west) -- ([yshift=\YSHIFT2, xshift=-\XSHIFT2]adapter3.north east) node[label, align=center] {Adapter III\\(20 nt)};
    \draw[brace] ([yshift=\YSHIFT2, xshift=\XSHIFT2]output.north west) -- ([yshift=\YSHIFT2, xshift=-\XSHIFT2]output.north east) node[label, align=center] {Output\\(42 nt)};
    \draw[brace] ([yshift=\YSHIFT2, xshift=\XSHIFT2]adapter4.north west) -- ([yshift=\YSHIFT2, xshift=-\XSHIFT2]adapter4.north east) node[label, align=center] {Adapter IV\\(21 nt)};
    \draw[brace] ([yshift=\YSHIFT2, xshift=\XSHIFT2]input4.north west) -- ([yshift=\YSHIFT2, xshift=-\XSHIFT2]input4.north east) node[label, align=center] {Input IV\\(10 nt)};
    \draw[brace] ([yshift=\YSHIFT2, xshift=\XSHIFT2]adapter2.north west) -- ([yshift=\YSHIFT2, xshift=-\XSHIFT2]adapter2.north east) node[label, align=center] {Adapter II\\(21 nt)};
    \draw[brace] ([yshift=\YSHIFT2, xshift=\XSHIFT2]input2.north west) -- ([yshift=\YSHIFT2, xshift=-\XSHIFT2]input2.north east) node[label, align=center] {Input II\\(10 nt)};
    \draw[brace] ([yshift=\YSHIFT2]handle2.north west) -- ([yshift=\YSHIFT2]handle2.north east) node[label, align=center] {Handle II\\(20 nt)};

    \draw[red, thick, |-|] ([yshift=-0.3cm, xshift=-0.7cm]input1.south west) -- ([yshift=-0.3cm, xshift=-0.1cm]input1.south east);

    \draw[red, thick, |-|, yshift=-0.6cm] ([yshift=-0.3cm, xshift=0.1cm]input2.south west) -- ([yshift=-0.3cm, xshift=0.7cm]input2.south east);

    \draw[green, thick, |-|] ([yshift=-0.3cm, xshift=-0.7cm]input3.south west) -- ([yshift=-0.3cm, xshift=-0.1cm]input3.south east);

    \draw[green, thick, |-|, yshift=-0.6cm] ([yshift=-0.3cm, xshift=0.1cm]input4.south west) -- ([yshift=-0.3cm, xshift=0.7cm]input4.south east);

    \begin{scope}[xshift=-0.75cm, yshift=-1.5cm]
        \node[constant, minimum height=0.25cm, minimum width=0.5cm] (constant) {};
        \node[right, right=of constant] {Constant/determined part};
        \node[random, below=0.2cm of constant, minimum height=0.25cm, minimum width=0.5cm] (random) {};
        \node[right, right=of random] {Randomly synthesized part};
        \node[below=0.2cm of random, minimum height=0.25cm, minimum width=0.5cm] (selection1) {};
        \draw[red, thick, |-|] (selection1.west) -- (selection1.east);
        \node[right, right=of selection1] {1st Selection PCR Primer binding site};
        \node[below=0.2cm of selection1, minimum height=0.25cm, minimum width=0.5cm] (selection2) {};
        \draw[green, thick, |-|] (selection2.west) -- (selection2.east);
        \node[right, right=of selection2] {2nd Selection PCR Primer binding site};
    \end{scope}
\end{tikzpicture}
    \caption{The structure of a two-stage \ac{ordna} molecule.}
    \label{fig:2-stage_ordna_design}
\end{figure*}
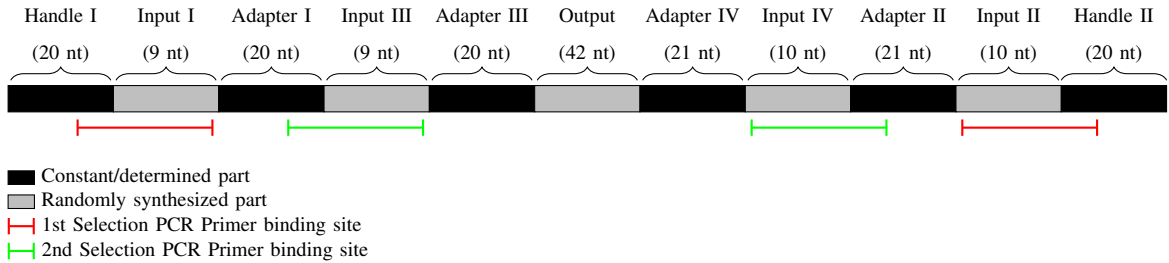

If we assumed that this scheme can be extended arbitrarily, the generation function $\gen$ would take the security parameter $n$ as input and set the number of selection \ac{pcr} stages as a linear function in $n$ to generate the \ac{ordna}. 
Furthermore, the length of the random output part will also increase linearly in $n$. This will be adopted by the MinHash algorithm and the error-correcting code in the extraction algorithm, such that the length and dimension, as well as the minimum distance of the code, also increase linearly in $n$. Therefore, the number of possible outputs will increase exponentially in $n$.
In the following we refer to this scheme as \emph{multi-stage \ac{ordna}}. We show that the multi-stage \ac{ordna} is asymptotically unclonable in \Cref{sec:asymptotic_unpredictability}.
\end{example}

\subsection{Unpredictability}\label{sec:unpredictability}

\Acp{cf} can also serve as a source of shared secrets. Owing to the divisibility, multiple replicas of the same underlying randomness can be produced and distributed. Possessing two pools derived from the same \ac{cf} enables two parties to establish a common secret: both issue an identical challenge and, using publicly shared helper data, extract the same secret. For this construction to be secure, the \ac{cf} must be unpredictable.

Unpredictability differs from unclonability: an adversary should be unable to compute the correct response to a fresh challenge even without constructing a physical clone. For instance, a scheme that is straightforward to analyze but prohibitively expensive to replicate chemically might satisfy asymptotic unclonability while failing to provide unpredictability. We therefore introduce a distinct security game that models an adversary's task of predicting the response to a new challenge after interacting with the authentic \ac{cf}; it coincides with the cloning game, defined in \Cref{def:cloning_game}, except for the final step.

\begin{definition}[Prediction Game]\label{def:prediction_game}
    For a \acl{cfs} $\cfs{\profile}$ and an adversary $\adversary$, the \emph{prediction game} $\sattack_{\adversary, \cfs{\profile}}(n)$ is defined as follows.
    \begin{itemize}
        \item An authentic \acl{cfs} $\cfs{\profile}$ with profile $\profile \in \setprofile$ is generated via $\gen(n)$.
        \item The adversary $\adversary$ is given oracle access to the \acl{cf} $\cf{\profile}$ of $\cfs{\profile}$ and may adaptively query challenges $x_1, \ldots, x_q \in \setchal$, receiving responses $y_1, \ldots, y_q$. Let $\mathcal{Q}_{\adversary} = \{x_1, \ldots, x_q\}$ denote the set of queried challenges.
        \item A fresh challenge $x \notin \mathcal{Q}_{\adversary}$ is chosen uniformly at random and some helper data $h$ is obtained through $(z, h) = \cfs{\profile}(x, \emptyset)$. If the adversary predicts the output $z$ that is produced by $(z,h) = \cfs{\profile}(x,h)$, the outcome of the game is $\sattack_{\adversary, \cfs{\profile}}(n) = 1$.
    \end{itemize}
\end{definition}

Again, we analyze unpredictability in a finite domain and asymptotically, similar to unclonability in \Cref{sec:unclonability}.

\subsubsection{Finite Domain Unpredictability}

\begin{definition}[Finite Domain Unpredictability]\label{def:finite_domain_unpredictability}
    Let $\cfs{\profile}$ be a \acl{cfs} with profile $\profile~\in~\setprofile$.
    Let an adversary $\adversary$ play the prediction game of \Cref{def:prediction_game} on $\cfs{\profile}$, issuing at most $q$ (adaptive) challenge queries and thus obtaining $q$ responses. 
    We say that $\cfs{\profile}$ is $(\sigma, q)$-\emph{unpredictable} with respect to $\setchal$ if 
    \begin{align*}
        \prob{\sattack_{\adversary, \cfs{\profile}}(n) = 1} \leq \sigma
    \end{align*}
\end{definition}

The adversary's success probability in predicting the correct response to a fresh challenge is bounded by the robustness $\rho(\cfs{\profile}, x)$ of the \ac{cfs} for that challenge.

\begin{example}
    Continuing the example from \Cref{sec:unclonability}, we analyze the unpredictability of the \ac{ordna} scheme \cite{luescher_2024_ChemicalUnclonablea} in a finite domain manner. In the trivial case, when $q=0$, an adversary is left with guessing the correct response. Since the extraction algorithm maps the responses to a codeword of a Reed-Solomon code with dimension $k=32$, the adversary has to guess one out of $256^{32}$ codewords. Therefore, the success probability is given by $\sigma = 1/256^{32} \approx 8.6 \times 10^{-78}$ and thus the \ac{ordna} scheme is $(8.6 \times 10^{-78}, 0)$-unpredictable. In the other extreme, if the adversary has queried all possible challenges, the correct response can simply be looked up. However, the success probability is still upper-bounded by the system's robustness. Therefore, the scheme is $(\rho({\cfs{\profile}, x}), 4^{13})$-unpredictable.
    
    In the non-trivial case, we bound the success probability of the adversary's attack with the input entropy of the \ac{ordna} pool for disjoint sets of challenges. Based on the practical observations that challenges with only one differing nucleotide are slightly correlated we exaggerate the capabilities of the adversary to know everything about similar challenges than the ones observed and nothing about the remaining challenges. 
    \begin{theorem}\label{the:ordna_finite_domain_unpredictability}
        Let $\kappa$ be a positive integer. Given a challenge $x' \in \setchal$, we denote by $\mathcal{B}_{\kappa}(x')$ the Hamming ball of radius $\kappa$ centered at $x'$. Assume that no information is leaked for challenges $x \notin \mathcal{B}_{\kappa}(x')$. Then, given the length $\ell$ of the selection \ac{pcr} primers, the \ac{ordna} scheme is $(\sigma, q)$-unpredictable with 
        \begin{align*}
            \sigma \geq \frac{1 + q \sum_{i=0}^{\kappa} \binom{\ell}{i} \cdot 3^{i}}{4^{\ell}}.
        \end{align*}
    \end{theorem}
    \begin{proof}
        Assume an adversary $\adversary$ performs the prediction game $\sattack_{\adversary, \cfs{\profile}}(n)$. Recall that $\mathcal{Q}_{\adversary}$ denotes the set of messages queried by $\adversary$, and denote by $\mathcal{P}_{\adversary}~=~\{(x_1,z_1), \ldots, (x_q,z_q)\}$ the set of challenge-output pairs.
        Let $x \in \setchal = \set{\A, \C, \G, \T}^{\ell}$ be a challenge chosen uniformly at random, and let $Z$ be a random variable representing the output of the \ac{ordna} scheme for $x$.
        We overestimate the capabilities of $\adversary$ such that $\adversary$'s attack is successful if for the guessed new challenge $x$ there is an $x_i \in \mathcal{Q}_{\adversary}$ such that $x \in \mathcal{B}_{\kappa}(x_i)$. Additionally, the attack may be successful if $\adversary$ guesses $x$ correctly even though it does not lie in any $\kappa$-neighborhood. For each $x_i \in \mathcal{Q}_{\adversary}$, we denote the set of neighboring challenge-output pairs by $\mathcal{N}_{\kappa}(x_i) = \{(x', z) \mid (z,h) = \cfs{\profile}(x', h),\ x' \in \mathcal{B}_{\kappa}(x_i)\}$. Note that, for $x_i \in \mathcal{Q}_{\adversary}$ and for all $x \in \mathcal{N}_{\kappa}(x_i)$ we have $H(Z \ | \ x) \geq 0$, which allows us to bound
        \begin{align}
            \prob{\sattack_{\adversary, \cfs{\profile}}(n) = 1} 
            \leq \prob{x \in \bigcup_{i=1}^{q} \mathcal{N}_{\kappa}(x_i)} 
            \nonumber
            \\
            + \prob{x \notin \bigcup_{i=1}^{q} \mathcal{N}_{\kappa}(x_i)} \cdot 2^{-H(Z \ | \ x \notin \mathcal{Q}_{\adversary}, \mathcal{P}_{\adversary})} \label{eq:prob_prediction_win}
        \end{align}
        For any $x_i \in \setchal$, we denote $\Phi_i(\ell,\kappa) := |\mathcal{N}_{\kappa}(x_i)|$. Note that, $\Phi_i(\ell,\kappa) = |\mathcal{B}_{\kappa}(x_i)|$ which is independent of the center $x_i \in \setchal$. Indeed, for any $i = 1, \ldots, q$, we have
        \begin{align*}
            \Phi_i(\ell, \kappa) = \sum_{j=0}^{\kappa} \binom{\ell}{j} \cdot (4-1)^{j}.
        \end{align*}
        For simplicity, we thus omit the subscript $i$ and only write $\Phi(\ell, \kappa)$.
        Then, by the union bound, the probability of $x$ lying in one of the neighborhoods $\mathcal{N}_{\kappa}(x_i)$ is given by
        \begin{align*}
            \prob{x \in \bigcup_{i=1}^{q} \mathcal{N}_{\kappa}(x_i)} \leq \frac{q \Phi(\ell, \kappa)}{4^{\ell}}.
        \end{align*}
        We can now bound the conditional entropy $H(Z | x \notin~\mathcal{Q}_{\adversary}, \mathcal{P}_{\adversary})$ from below as
        \begin{align*}
            H(Z | x \notin \mathcal{Q}_{\adversary}, \mathcal{P}_{\adversary}) 
            &\geq H\left(Z \ \bigg| \ x \in \setchal \setminus \bigcup_{i=1}^{q} \mathcal{N}_{\kappa}(x_i)\right) 
            \\
            &\approx \log_2(|\setchal| - q \Phi(\ell, \kappa)).
        \end{align*}
        Note that the inequality follows from the fact that even within $\mathcal{B}_{\kappa}(x_i)$ the adversary does not necessarily have full information.
        Eventually, by substituting these values in \Cref{eq:prob_prediction_win}, we get
        \begin{align*}
            & \prob{\sattack_{\adversary, \cfs{\profile}}(n) = 1} \\
            & \quad \leq \frac{q \Phi(\ell, \kappa)}{4^{\ell}} + \left(1 - \frac{q \Phi(\ell, \kappa)}{4^{\ell}}\right) \frac{1}{4^{\ell} - q \Phi(\ell, \kappa)} \\
            & \quad = \frac{q \Phi(\ell, \kappa) + 1}{4^{\ell}}
        \end{align*}
        which concludes the proof.

    \end{proof}
    The assumption that an adversary can predict any $x \in \mathcal{B}_{\kappa}(x_i)$ is based on the observations from the data in \cite{luescher_2024_ChemicalUnclonablea}, showing that the response to challenges which differ in one position are not completely uncorrelated. 
    It suggests that the average relative Hamming distance of the responses of $x, x'$ with $\dist{x,x'} = 1$ after the MinHash algorithm is approximately $89.25\%$ which is less than the expected distance of two random vectors given by $1 - \frac{1}{256} \approx 99.61\%$. Therefore, we are overestimating the adversary's capabilities which makes the derived bound on the success probability more conservative. 
    Additionally, in \cite{luescher_2024_ChemicalUnclonablea} they chose $\ell = 13$ and take $Z$ to be the random variable with realizations in a Reed-Solomon code $\mathcal{C}_{RS}$ of length $n = 255$ and dimension $k = 32$. Since then $|\setchal| = 4^{13} \ll 256^{32} = |\mathcal{C}_{RS}|$, it is very unlikely that two different challenges map to the same codeword. Thus, we can assume that $H(Z)~\approx~\log_2(4^{13})~=~26$ bits.
    
    \begin{table}
        \centering
        \begin{tabular}{c|ccc}
            & $\kappa = 1$ & $\kappa = 2$ & $\kappa = 3$
            \\ \hline
            $q = 10^0$ & $6.1 \times 10^{-7}$ & $1.1 \times 10^{-5}$ & $1.3 \times 10^{-4}$ 
            \\
            $q = 10^1$ & $6.1 \times 10^{-6}$ & $1.1 \times 10^{-4}$ & $1.3 \times 10^{-3}$ 
            \\
            $q = 10^2$ & $6.1 \times 10^{-5}$ & $1.1 \times 10^{-3}$ & $1.3 \times 10^{-2}$ 
            \\
            $q = 10^3$ & $6.1 \times 10^{-4}$ & $1.1 \times 10^{-2}$ & $1.3 \times 10^{-1}$ 
            \\
            $q = 10^4$ & $6.1 \times 10^{-3}$ & $1.1 \times 10^{-1}$ & $1$ 
            \\
            $q = 10^5$ & $6.1 \times 10^{-2}$ & $1$ & $1$ \\
        \end{tabular}
        \caption{Bound on an adversary's success probability $\sigma$ in predicting the correct response to a fresh challenge for different values of $\kappa$ and number of queries $q$.}
        \label{tab:unpredictability_success_probabilities}
    \end{table}

    In \Cref{tab:unpredictability_success_probabilities}, we present the derived bound on an adversary's success probability $\sigma$ in predicting the correct response to a fresh challenge for different values of $\kappa$ and number of queries $q$. Note that the bound only works for small values of $q$ and $\kappa$ as otherwise the success probability will exceed $1$. However, this assumption is plausible as in the attack scenario, we do not expect an adversary to have access to $\cf{\profile}$ itself but can only query a limited number of challenges. 
\end{example}

\subsubsection{Asymptotic Unpredictability}\label{sec:asymptotic_unpredictability}

For defining unpredictability in an asymptotic regime, we again follow the concepts of \Cref{subsec:authentication} and \cite{katz_2020_IntroductionModern}. Asymptotic unpredictability is then defined as the inability of any adversary to win the prediction game $\sattack_{\adversary, \cfs{}}(n)$ with non-negligible probability in $n$. 

\begin{definition}[Asymptotic Unpredictability]\label{def:asymptotic_unpredictability}
    A \acl{cfs} $\cfs{\profile}$ is \emph{asymptotically unpredictable} if for all polynomial-time adversaries $\adversary$, it holds that
    \begin{align}\label{eq:unpredictability}
        \prob{\sattack_{\adversary, \cfs{\profile}}(n) = 1} \in \negl(n)
    \end{align}
\end{definition}

\begin{example}
    We will now analyze the asymptotic unpredictability of the multi-stage \ac{ordna} scheme introduced in \Cref{ex:ordna_asymptotic_unclonability}. Similarly to the finite domain case, we can bound the success probability of an adversary with the entropy of the \ac{ordna} pool.
    \begin{theorem}\label{the:ordna_asymptotic_unpredictability}
        The multi-stage \ac{ordna} scheme is asymptotically unpredictable.
    \end{theorem}
    \begin{proof}
        From \Cref{the:ordna_finite_domain_unpredictability}, the success probability of a prediction attack is bounded by
        \begin{align*}
            \prob{\sattack_{\adversary, \cfs{\profile}}(n) = 1} 
            \leq \frac{1 + q \sum_{i=0}^{\kappa} \binom{\ell}{i} \cdot 3^{i}}{4^{\ell}},
        \end{align*}
        with $\ell$ being the combined length of the selection \ac{pcr} primers and therefore a linear function in $n$. As a prediction attack is restricted to polynomial time complexity in $n$, $q$ is a polynomial function in $n$ too. There is no reason to assume that the cross interference in the selection \ac{pcr} will increase with the number of stages. Therefore, we can assume that $\kappa$ is a constant. 
        By bounding the binomial coefficient with $\binom{\ell}{i} \leq \frac{\ell^{i}}{i!}$, we get
        \begin{align*}
            \prob{\sattack_{\adversary, \cfs{\profile}}(n) = 1} 
            &\leq 
            \frac{1 + q \sum_{i=0}^{\kappa} (3\ell)^{i} \frac{1}{i!}}{4^{\ell}} 
        \end{align*}
        which is negligible with respect to $n$, since the denominator is an exponential function in $n$ while the numerator only consists of polynomial functions in $n$. Therefore, the multi-stage \ac{ordna} scheme is asymptotically unpredictable.
    \end{proof}
\end{example}

In the asymptotic setting, unpredictability is a stronger notion than unclonability. If an adversary can generate a clone that behaves like the authentic \ac{cf} for a fresh challenge, it can also predict the correct response to that challenge. Therefore, asymptotic unpredictability implies asymptotic unclonability as stated in the following proposition.

\begin{proposition}\label{prop:strong_implies_weak_security}
    If a \acl{cfs} is asymptotically unpredictable, then it is also asymptotically unclonable.
\end{proposition}
\begin{proof}
    Assume by contradiction that for every $n$ there exists \acl{cfs} $\cfs{\profile}$ with profile $\profile$ that is unpredictable but clonable. Then there exists an efficient attack that can construct $\cfs{\profile'}$ such that for a fresh challenge $x$ it holds $\prob{\attack_{\adversary, \cfs{\profile}}(n) = 1}~\notin~\negl(n)$. This implies that 
    \begin{align*}
        \prob{\cfs{\profile'}(x,h) = \cfs{\profile}(x,h)} \notin \negl(n).
    \end{align*}
    The probability is over a fresh challenge $x$ with corresponding helper data obtained by $(x,h) = \cfs{\profile}(x, \emptyset)$. Possessing such a \ac{cfs} with profile $\profile'$, an adversary can answer the fresh challenge with non-negligible probability in $n$, thereby also predicting the correct response in the prediction game (since evaluating $\cfs{\profile'}$ is efficient). This contradicts the assumed security against prediction.
\end{proof}

However, in case the \ac{cf} is based on \ac{dna} it is possible that a \ac{cfs} is asymptotically unpredictable but amplifiable. As mentioned before, we need to distinguish between \emph{Sequencing and Synthesis} attacks and \emph{Chemical Amplification} attacks. The cloning and unpredictability games (see \Cref{def:cloning_game} and \Cref{def:prediction_game}, respectively) give an adversary only the power for a \emph{Sequencing and Synthesis} attack. If an adversary had the power to amplify the \ac{cf} chemically, it would be possible to have a clonable but unpredictable \ac{cfs} without contradicting \Cref{prop:strong_implies_weak_security}. \Cref{the:ordna_asymptotic_unpredictability} and \Cref{prop:strong_implies_weak_security} immediately imply the following result.

    \begin{corollary}
        The multi-stage \ac{ordna} scheme is an asymptotically unclonable \ac{cfs}.
    \end{corollary}

We could show results on the unclonability and unpredictability of the \ac{ordna} scheme in both finite domain and asymptotic settings. 
These results demonstrate the potential of \ac{dna}-based chemical functions to provide strong security guarantees based on the inherent complexity of molecular interactions. On the other hand, it will not be possible to show similar results for the simple dye molecule scheme from \Cref{ex:dye_extract} as the number of possible challenges is limited and the concentrations can easily be measured and reproduced.

\section{Practical Limitations}\label{sec:limitations}
While the proposed framework provides a unified and formally grounded view on chemical functions, several practical aspects lie outside its scope and deserve explicit acknowledgment.
\subsection{Cost-based attack analysis}
Due to the rapidly evolving nature of the field, we do not perform a cost analysis for any specific attack vector. The economic feasibility of attacks against chemical functions depends strongly on the fast-changing costs of DNA synthesis, sequencing, and laboratory equipment, as well as on the specific wet-lab realization of the primitive. Rather than fixing such numbers, we evaluate the schemes through a security parameter $n$ that can be adapted to technological advancements.
\subsection{ML-based modeling attacks}
We build the framework from an information-theoretic perspective, so that security does not rely on any computational assumption. The unpredictability notion measures the residual randomness of responses given any prior view of the system and therefore also bounds the success probability of attacks based on machine learning algorithms: beyond the guarantee captured by the security parameter, there is no structure left to be learned. In a concrete deployment, however, the actual security can be weaker than predicted by the analysis. Bias in the chemical processes, an unfortunate choice of extraction function, or other implementation flaws could embed structure into the challenge--response behavior that an ML-based modeling attack could exploit. Mitigating this requires careful design and characterization of both the chemical primitive and the digital post-processing pipeline, which we leave to future work.
\subsection{Masquerading and denial-of-service attacks}
The authentication pipeline for chemically unclonable functions relies on digitized challenges, measured responses, and decision logic, exposing both computational and physical attack surfaces. At the computational and protocol layers, chemically unclonable functions share vulnerabilities with \acp{puf}, where poorly secured challenge--response exchanges can be exploited for masquerading attacks or computational resource exhaustion \cite{alyanbaawi_2025_MITMDoSResistant,lounis_2023_LessonsLearned}.
However, chemically unclonable functions also introduce unique physical-layer threat vectors. For masquerading, deliberate noise injection at the measurement stage could shift counterfeit responses toward the acceptance thresholds of legitimate samples. For denial-of-service, if the chemical primitive cannot be amplified, the system is susceptible to material exhaustion. While the divisible nature of chemical functions prevents instant destruction from a single localized physical attack, an adversary could theoretically mount a physical denial-of-service attack by repeatedly querying the primitive until the manufacturer's chemical material is entirely depleted.
Physical masquerading attacks can in principle be mitigated by a careful design of the acceptance thresholds of the extraction algorithm and measurement procedure, though a quantitative treatment requires assumptions on the attacker's access to the physical device. Furthermore, while physical-layer attacks require domain-specific threshold tuning, computational vulnerabilities require rigorous protocol-level countermeasures, such as cryptographically obfuscated challenge--response pairs or sensor-initiated handshakes \cite{alyanbaawi_2025_MITMDoSResistant,lounis_2023_LessonsLearned}. All of these aspects are beyond the scope of this work and are left for future research.
\subsection{DNA amplification attacks}
We exclude direct DNA amplification attacks from the mathematical analysis. Amplification is one of the most natural attack vectors against DNA-based chemical functions, and any unclonability claim should be read under the assumption that it is counteracted by chemical means, for instance through protective chemistry or design choices that render amplification infeasible. The security guarantees established in this paper therefore hold conditionally on such a protection mechanism being in place.

\section{Applications of Chemical Functions}\label{sec:application}
\aclp{cf} could have a wide range of potential applications. In the following sections, we will explore two specific applications of \acp{cf} in the authentication of many different materials and key generation.
\subsection{Chemical Authentication Schemes}\label{sec:cf_auth}

Chemical functions provide a foundation for designing authentication schemes. At a high level, an instance of a \acl{cf} $\cf{\profile}$ with profile $\profile$ (henceforth, a ``valid'' reference instance) and another instance $\cf{\profile'}$ with profile $\profile'$ of unknown origin are stimulated with the same challenge $x$. Each evaluation produces a noisy response due to the inherent stochasticity of chemical processes and measurement. A suitable extraction procedure then maps these noisy responses to stable representatives that can be compared by a verification algorithm. If the representatives match under the verifier, the scheme accepts; otherwise, it rejects.

Using evaluation parameters $\alpha_{\mathrm{E}}$, a challenge $x$ is applied to $\cf{\profile}$ and $\cf{\profile'}$ yielding noisy responses $y$ and $y'$. An extraction algorithm in setup mode (empty helper data $\emptyset$ as input) parameterized by $\alpha_{\mathrm{X}}$ derives a stable output $z$ and helper data $h$ from $y$. The same extraction algorithm using the same extraction parameters $\alpha_{\mathrm{X}}$, when applied to $y'$ with helper data $h$, produces an output $z'$. A verification algorithm $\verify$ finally compares $z$ and $z'$ and outputs $1$ (accept) if they are close in some chosen metric, and $0$ (reject) otherwise. In practice, $\verify$ can be an exact equality test or a distance-threshold test on multiple outputs from each \ac{cf}. \Cref{fig:cf_authentication} summarizes the protocol.
Note that the extraction algorithm could also be the identity function, i.e., no extraction is performed. In that case the verification algorithm would compare the distance of the two noisy responses directly and decide acceptance based on a threshold.

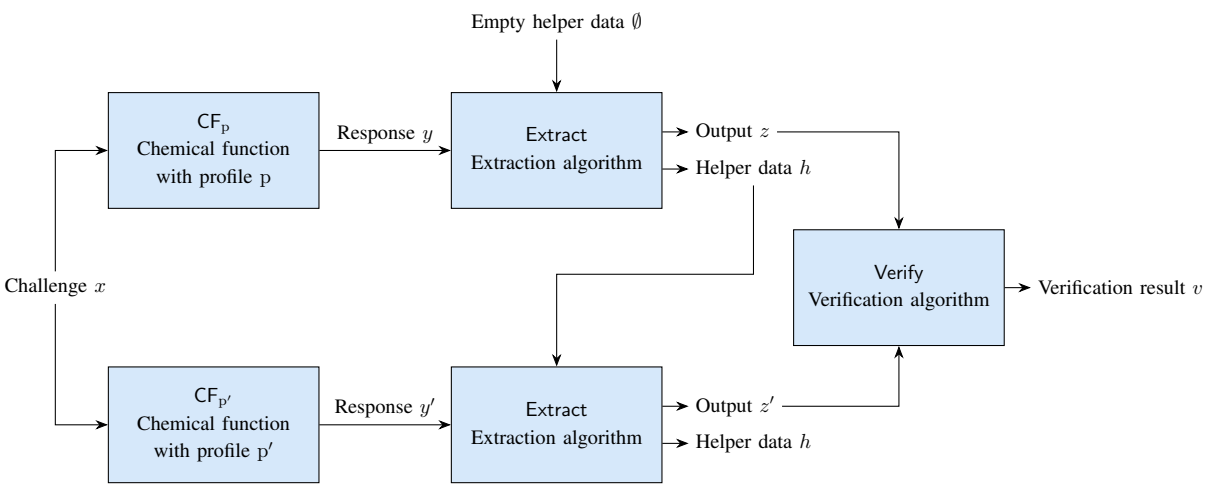
\begin{figure*}[ht]
    \centering
    \begin{tikzpicture}[node distance=0cm, transform shape, scale=0.7]

\definecolor{noisecolor}{HTML}{5c0404} 

\tikzset{
    function/.style={
        rectangle, 
        draw, 
        fill=bleudefrance!20, 
        align=center, 
        inner ysep=10pt,
        minimum width=4cm,
        minimum height=2.2cm,
        execute at begin node=\setlength{\baselineskip}{15pt}
    },
    component/.style={
        rectangle, 
        draw=green!30!black, 
        fill=green!15!white, 
        dashed, 
        align=center, 
        inner ysep=10pt,
        minimum width=4cm,
        minimum height=1.2cm,
        execute at begin node=\setlength{\baselineskip}{15pt}
    },
    arrow_style/.style={
        ->, 
        >=Stealth
    },
    param_arrow/.style={
        ->, 
        >=Stealth, 
        dashed, 
        draw=green!30!black
    },
    label_style/.style={ 
        execute at begin node=\setlength{\baselineskip}{15pt}
    },
    system_box/.style={
        rectangle, 
        draw, 
        thick, 
        inner sep=10pt
    }
}

\node (chemical_function) [function] {$\cf{\profile}$ \\ Chemical function \\ with profile $\profile$};
\node (extract_func) [function, right=2.5cm of chemical_function] {$\extract$ \\ Extraction algorithm};

\node (chemical_function2) [function, below=3cm of chemical_function] {$\cf{\profile'}$ \\ Chemical function \\ with profile $\profile'$};
\node (extract_func2) [function, right=2.5cm of chemical_function2] {$\extract$ \\ Extraction algorithm};

\node (compare) [function] at ($(extract_func)!0.5!(extract_func2) + (6.5cm,0)$) {$\verify$ \\ Verification algorithm};

\node (challenge_x) at ($(chemical_function.west)!0.5!(chemical_function2.west) + (-1cm,0)$) {Challenge $x$};

\node (helper_data_h1_in) [above=1cm of extract_func.north] {Empty helper data $\emptyset$};
\node (helper_data_h1_out) [right=0.5cm of extract_func.east, yshift=-0.5\baselineskip] {Helper data $h$};
\node (output_z) [right=0.5cm of extract_func.east, yshift=0.5\baselineskip] {Output $z$};
\node (output_z2) [right=0.5cm of extract_func2.east, yshift=0.5\baselineskip] {Output $z'$};
\node (helper_data_h2_out) [right=0.5cm of extract_func2.east, yshift=-0.5\baselineskip] {Helper data $h$};


\draw [arrow_style] (challenge_x.north) |- (chemical_function.west);
\draw [arrow_style] (challenge_x.south) |- (chemical_function2.west);
\draw [arrow_style] (chemical_function.east) -- node[above] {Response $y$} (extract_func.west);
\draw [arrow_style] (chemical_function2.east) -- node[above] {Response $y'$} (extract_func2.west);

\draw [arrow_style] (helper_data_h1_in) -- (extract_func.north);
\draw [arrow_style] ([yshift=-0.5\baselineskip] extract_func.east) -- (helper_data_h1_out);
\draw [arrow_style] ([yshift=0.5\baselineskip] extract_func.east) -- (output_z);
\draw [arrow_style] (helper_data_h1_out) -- ++(0,-2cm) -| (extract_func2.north);
\draw [arrow_style] ([yshift=0.5\baselineskip]extract_func2.east) -- (output_z2);
\draw [arrow_style] ([yshift=-0.5\baselineskip] extract_func2.east) -- (helper_data_h2_out);

\draw [arrow_style] (output_z) -| (compare.north);
\draw [arrow_style] (output_z2) -| (compare.south);
\draw [arrow_style] (compare.east) -- ++(0.5,0) node[right] {Verification result $v$};

\end{tikzpicture}
    \caption{Illustration on how chemical functions can be used for authentication. For clarity, the evaluation parameter $\alpha_{\text{E}}$ and the extraction parameter $\alpha_{\text{X}}$ are omitted. The challenge is sent to both chemical functions and results in two noisy responses $y$ and $y'$. The extraction algorithm is applied to the first response with empty helper data to generate the output $z$ and helper data $h$. The same extraction algorithm is applied to the second response but this time with the helper data $h$ to generate the output $z'$. Finally, both outputs are compared using a verification algorithm. The verification algorithm outputs $1$ if both outputs are closer than a defined threshold and $0$ otherwise.}
    \label{fig:cf_authentication}
\end{figure*}

Crucially, unclonability underpins the soundness of this inference: if an adversary cannot fabricate a distinct function that reproduces the responses of the valid instance, then matching outputs imply, with high probability, that the presented function shares the same profile as the reference, i.e., $\profile = \profile'$. Without unclonability, an attacker might construct a clone that matches $z$ on selected challenges, undermining authenticity. Hence, alongside robustness, unclonability is essential to guarantee that acceptance by $\verify$ evidences possession of the specific chemical function instance rather than a counterfeit.

Using this scheme, it is possible to realize an authentication flow for many different materials. Chemical functions can be designed at the molecular level as shown in \cite{luescher_2024_ChemicalUnclonablea}. By encapsulating them in silica particles as done in \cite{grass_2015_RobustChemical}, it is feasible to embed the \ac{cf} in the material. Applications range from protecting textiles and food with problematic origins to pharmaceuticals and critical components in high-stakes industries. An analysis of potential applications can be found in \cite{kuzdralinski_2023_UnlockingPotential, didax-project_2024_StateArt}.

\subsection{Key Generation}\label{sec:key_generation}
Similar to \acp{puf}, \acp{cf} can be used for key generation. The core idea is to treat the noisy response of a \ac{cf} as a source of randomness from which a cryptographic key is derived. With an appropriate extraction algorithm, a robust key can be reproduced from noisy responses. Only the helper data needs to be stored. Because this helper data is designed not to reveal information about the key, it can be kept in unsecured memory. When the key is required, the \ac{cf} is evaluated on the same challenge and, together with the helper data, the same key is reconstructed; the key itself never needs to be stored.

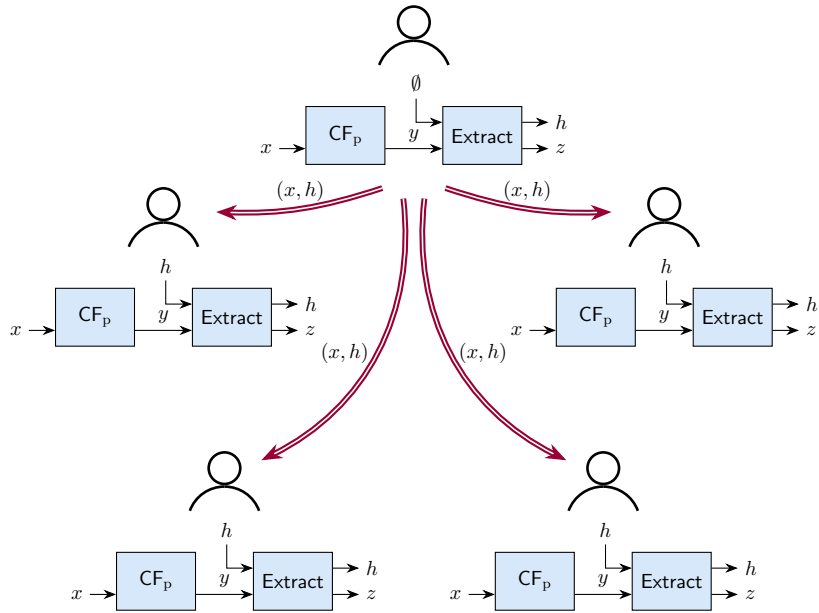
\begin{figure*}[ht]
    \centering
    \begin{tikzpicture}[node distance=0cm, transform shape, scale=0.7]

\definecolor{noisecolor}{HTML}{5c0404} 

\tikzset{
    function/.style={
        rectangle, 
        draw, 
        fill=bleudefrance!20, 
        align=center, 
        inner ysep=10pt,
        minimum width=1.5cm,
        minimum height=1.1cm,
        execute at begin node=\setlength{\baselineskip}{15pt}
    },
    arrow_style/.style={
        ->, 
        >=Stealth
    }, 
    dist_arrow/.style={
        ->, 
        >=Stealth, 
        thick, 
        double, 
        draw=purple!80!black
    },
    pics/person/.style={
        code={
            \draw[black, line width=1pt] (0,0) arc (90:20:0.7);
            \draw[black, line width=1pt] (0,0) arc (90:160:0.7);
            \draw[black, line width=1pt] (0,0.35) circle (0.3);
        }
    },
    pics/cfs/.style args ={#1}{
        code={
            \pic (person) at (0,0.8) {person};

            \node (chemical_function) at (-1.3,-1) [function] {$\cf{\profile}$};
            \node (extract_func) at (1.3,-1) [function] {$\extract$};

            \node (challenge_x) [left=0.5cm of chemical_function.west, yshift=-0.25cm] {$x$};

            \node (helper_data_h) [align=center, above=0.7cm of extract_func.west, xshift=-0.5cm] {#1};
            \node (helper_data_h2) [right=0.5cm of extract_func.east, yshift=0.25cm] {$h$};

            \node (output_z) [right=0.5cm of extract_func.east, yshift=-0.25cm] {$z$};

            \draw [arrow_style] (challenge_x) -- (challenge_x.east -| chemical_function.west);

            \draw [arrow_style] ([yshift=-0.25cm] chemical_function.east) -- node[midway, above] {$y$} ([yshift=-0.25cm] extract_func.west);

            \draw [arrow_style] (helper_data_h) |- ([yshift=0.25cm] extract_func.west);
            \draw [arrow_style] ([yshift=0.25cm] extract_func.east) -- (helper_data_h2);

            \draw [arrow_style] ([yshift=-0.25cm] extract_func.east) -- (output_z);
        }
    },
}

\pgfmathsetlengthmacro{\radius}{5cm}   
\coordinate (p1coord) at (90:\radius);
\pic (p1pic) at (p1coord) {cfs={$\emptyset$}};

\coordinate (p2coord) at (162:\radius);
\pic (p2pic) at (p2coord) {cfs={$h$}};

\coordinate (p3coord) at (224:\radius);
\pic (p3pic) at (p3coord) {cfs={$h$}};

\coordinate (p4coord) at (316:\radius);
\pic (p4pic) at (p4coord) {cfs={$h$}};

\coordinate (p5coord) at (18:\radius);
\pic (p5pic) at (p5coord) {cfs={$h$}};

\draw [dist_arrow] (p1coord) +(-0.6,-2) to[bend left=10] node[midway, above] {$(x,h)$} ($(p2coord) + (1,1)$);
\draw [dist_arrow] (p1coord) +(-0.2,-2.2) to[bend left=35] node[midway, left] {$(x,h)$} ($(p3coord) + (0.7,1.3)$);
\draw [dist_arrow] (p1coord) +(0.2,-2.2) to[bend right=35] node[midway, right] {$(x,h)$} ($(p4coord) + (-0.7,1.3)$);
\draw [dist_arrow] (p1coord) +(0.6,-2) to[bend right=10] node[midway, above] {$(x,h)$} ($(p5coord) + (-1,1)$);

\end{tikzpicture}
    \caption{Distributed key generation using \acp{cf}: multiple parties holding instances with practically identical profiles evaluate the same challenge and shared helper data to reconstruct the same secret key.}
    \label{fig:cf_distributed_key_generation}
\end{figure*}

For this application to be secure, the \ac{cf} must satisfy unpredictability, as discussed in \Cref{sec:unpredictability}. It should be computationally infeasible for an adversary to predict the response to a new challenge, even when given multiple responses to previous challenges.

Beyond the capabilities of \acp{pf}, \acp{cf} enable a distributed key-generation mechanism. A \ac{cf} can be divided into multiple instances with effectively identical profiles to be distributed to different parties. As illustrated in \Cref{fig:cf_distributed_key_generation}, each party can evaluate its instance on the same challenge and, using the same helper data, reconstruct the same key. Since the helper data reveals no information about the key, it can be shared over an insecure channel. This enables multiple parties to derive a shared secret on demand. Moreover, any party can independently initiate the protocol by choosing a fresh challenge, generating the corresponding helper data, and deriving a new key from the resulting output $z$.

\section{Genomic Sequence Encryption}\label{sec:genomic-sequence-encryption}
In this section, we show how \acf{gse}, as introduced in \cite{volf_2023_CryptographyDNA}, can be used to implement \acp{cf} and analyze their properties.
As mentioned in \Cref{sec:related}, the scheme is introduced as an encryption scheme, but an application in the context of \acp{cf} is also possible.
The challenge in this case consists of a set of genomic coordinates at which edits are introduced. The authors use \acf{abe} to convert $\A / \T$ to $\G / \C$ and \acf{cbe} to convert $\C / \G$ to $\T / \A$. The response is the respective binary sequence at these coordinates, where $0$ represents an unedited reference base and $1$ indicates a present base edit. As mentioned in \cite{volf_2023_CryptographyDNA}, sequencing is noisy and therefore the response is noisy as well.
Furthermore, the authors only edited a fraction of the \ac{dna} strands ($0.1\%$ for \ac{cbe} and $0.145\%$ for \ac{abe}). 
As in \cite{volf_2023_CryptographyDNA}, we refer to the change of bases of a single strand as \emph{edit} and to the position in the whole pool of cells where the specified base is edited in some cells as the \emph{key-site}.

To complete the definition of a \ac{cfs} with a robust output, a suitable extraction algorithm is needed. Since in \cite{volf_2023_CryptographyDNA} the authors did not utilize an extraction algorithm, we could use the \emph{Fuzzy Extractor} from \cite{dodis_2004_FuzzyExtractors} that is also used in \cite{luescher_2024_ChemicalUnclonablea}. However, instead of a Reed-Solomon code, a BCH code can be used, as the response is a binary string of fixed length.

\subsection{Robustness of GSE}

The experimental results in \cite{volf_2023_CryptographyDNA} indicate a noise model corresponding to a discrete memoryless binary symmetric channel 
with error probability $p_e = 0.036$. For the analysis, we proceed similarly to \Cref{ex:robustness}.
Let $X$ be a random variable indicating the number of erroneous positions among $n$ read positions. The robustness is then given by 
\begin{align*}
    \rho_{\text{avg}}(\cfi, \setchal) = \prob{X \leq t}
\end{align*}
where $t$ is the error correction capability of the code in the Fuzzy Extractor. Note that we assume $X~\sim~\bindis{n, p_e}$ and therefore robustness can be computed using the cumulative distribution function of the binomial distribution, i.e.,
\begin{align*}
    \rho_{\text{avg}}(\cfi, \setchal) = \sum_{i=0}^{t} \binom{n}{i} p_e^i (1-p_e)^{n-i}.
\end{align*}
Some numerical values for different values of $n$ and $t$ are given in \Cref{tab:gse_robustness}.

\begin{table*}
    \centering
    \begin{tabular}{c|ccccccc}
        $n \setminus \frac{t}{n}$ & $0.3$ & $0.4$ & $0.5$ & $0.6$ & $0.7$ & $0.8$ & $0.9$ \\ \hline
        $10$ & $5.00\times 10^{-2}$ & $9.87\times 10^{-3}$ & $1.38\times 10^{-3}$ & $1.35\times 10^{-4}$ & $8.67\times 10^{-6}$ & $3.33\times 10^{-7}$ & $5.77\times 10^{-9}$ \\
        $20$ & $2.19\times 10^{-2}$ & $1.33\times 10^{-3}$ & $3.86\times 10^{-5}$ & $5.30\times 10^{-7}$ & $3.19\times 10^{-9}$ & $7.10\times 10^{-12}$ & $3.80\times 10^{-15}$ \\
        $50$ & $1.95\times 10^{-3}$ & $3.90\times 10^{-6}$ & $1.10\times 10^{-9}$ & $4.45\times 10^{-14}$ & $2.25\times 10^{-19}$ & $1.00\times 10^{-25}$ & $1.54\times 10^{-33}$ \\
        $100$ & $4.09\times 10^{-5}$ & $3.03\times 10^{-10}$ & $3.94\times 10^{-17}$ & $9.88\times 10^{-26}$ & $3.81\times 10^{-36}$ & $1.15\times 10^{-48}$ & $4.74\times 10^{-64}$ \\
        $200$ & $2.26\times 10^{-8}$ & $2.39\times 10^{-18}$ & $6.78\times 10^{-32}$ & $6.62\times 10^{-49}$ & $1.49\times 10^{-69}$ & $2.11\times 10^{-94}$ & $6.18\times 10^{-125}$ 
    \end{tabular}
    \caption{Robustness values $\rho_{\text{avg}}(\cfi, \setchal)$ for different values of $n$ and $t$.}
    \label{tab:gse_robustness}
\end{table*}

\subsection{Unclonability of GSE}

Next, we look at the unclonability of the scheme. 
We may consider the challenge to be a subset of some locations in the genome from all possible locations. Then, a new challenge could include locations that are already observed during previous challenges. Like this, the goal for the adversary becomes a coupon collector problem which is beyond the scope of this work. 

Now we consider the scenario in which $\adversary$ performs full genome sequencing and can therefore obtain noisy information about all possible challenges.
Note that each challenge involves more than $1\,000$ reads per position which increases the effort of $\adversary$ significantly.
In this attack scenario, the adversary has different capabilities compared to what is defined in the cloning game in \Cref{def:cloning_game}. The open cloning game from \Cref{def:open_cloning_game} fits this scenario.
The final goal of the adversary is to obtain knowledge about the key-sites of the challenge, i.e., the positions in the genome that are edited.
To achieve this, we consider the case where an adversary $\adversary$ can obtain $q$ full genome reads. For simplicity of the analysis we assume that the ratio of \ac{abe} is the same as that of \ac{cbe} being $\prob{\text{Position $i$ is edited} \mid \text{Position $i$ is key-site}} = p_{\text{edit}} = 0.001$ (i.e., $0.1\%$) of all cells are edited. Let the error probability of sequencing be $p_{\text{seq-err}}$.

Given that a position $i$ is a key-site,
the probability for $\adversary$ to read an edit is given by 
\begin{align*}
    p_\text{detect} = p_{\text{edit}} (\overline{p}_{\text{seq-err}}) + (\overline{p}_{\text{edit}}) p_{\text{seq-err}}. 
\end{align*}
Here we denote $ \overline{p} = 1 - p$ as the complement of $p$. 
Hence, $\adversary$ may read an edit at position $i$ even though it was not edited. The probability to read $\ell$ edits is thus computed using a binomial distribution, i.e., 
\begin{align*}
    &\prob{\ell \text{ edits read} \mid \text{Position $i$ is key-site}} 
    \\ 
    &\hspace{4cm} = \binom{q}{\ell} p_{\text{detect}}^{\ell} (\overline{p}_{\text{detect}})^{q-\ell}
    \\
    &\prob{\ell \text{ edits read} \mid \text{Position $i$ is not key-site}} 
    \\
    &\hspace{4cm} = \binom{q}{\ell} p_{\text{seq-err}}^{\ell} (\overline{p}_{\text{seq-err}})^{q-\ell}
\end{align*}
The best choice for $\adversary$ is to employ a maximum likelihood decision for each position of the challenge. Thus, we define the following likelihood ratio
\begin{align*}
    \Theta = \frac{\prob{\text{Position $i$ is key-site} \mid \ell \text{ edits read}}}{\prob{\text{Position $i$ is not key-site} \mid \ell \text{ edits read}}}.
\end{align*}
The maximum likelihood decision is to decide for $i$ to be a key-site if $\Theta > 1$ and $i$ not being a key-site otherwise. Using Bayes' theorem we can rewrite $\Theta$ as
\begin{align*}
    \Theta = 
    \tfrac{\prob{\ell \text{ edits read} \mid \text{Position $i$ is key-site}}\cdot\prob{\text{Position $i$ is key-site}}}{\prob{\ell \text{ edits read} \mid \text{Position $i$ is not key-site}}\cdot\prob{\text{Position $i$ is not key-site}}}
\end{align*}

As the adversary wants to minimize the errors in the challenges, the probability that a position $i$ is a key-site relates to the challenge instead of to the whole genome. Therefore, we can assume $\prob{\text{Position $i$ is key-site}}~=~0.5$ and we get
\begin{align*}
    \Theta &= \frac{p_{\text{detect}}^{\ell} (\overline{p}_{\text{detect}})^{q-\ell}}{p_{\text{seq-err}}^{\ell} (\overline{p}_{\text{seq-err}})^{q-\ell}}
\end{align*}
For $\Theta = 1$ we get the following decision boundary:
\begin{align*}
    \frac{\ell}{q} &=
    \frac{\log
    \left(
        \frac{\overline{p}_{\text{seq-err}}}
        {\overline{p}_{\text{detect}}}
    \right)}
    { \log
    \left(
        \frac{p_{\text{detect}} \; \cdot \;\overline{p}_{\text{seq-err}} }
        {\overline{p}_{\text{detect}} \; \cdot \; p_{\text{seq-err}}}
    \right)}.
\end{align*}
A plot of the decision boundary is given in \Cref{fig:decision_boundary}. Note that this decision boundary is not limited to the attack scenario but also applies to the honest evaluation of the \ac{cf}. Therefore, to minimize the error probability every party needs to factor in the sequencing error rate of the used sequencing technology.
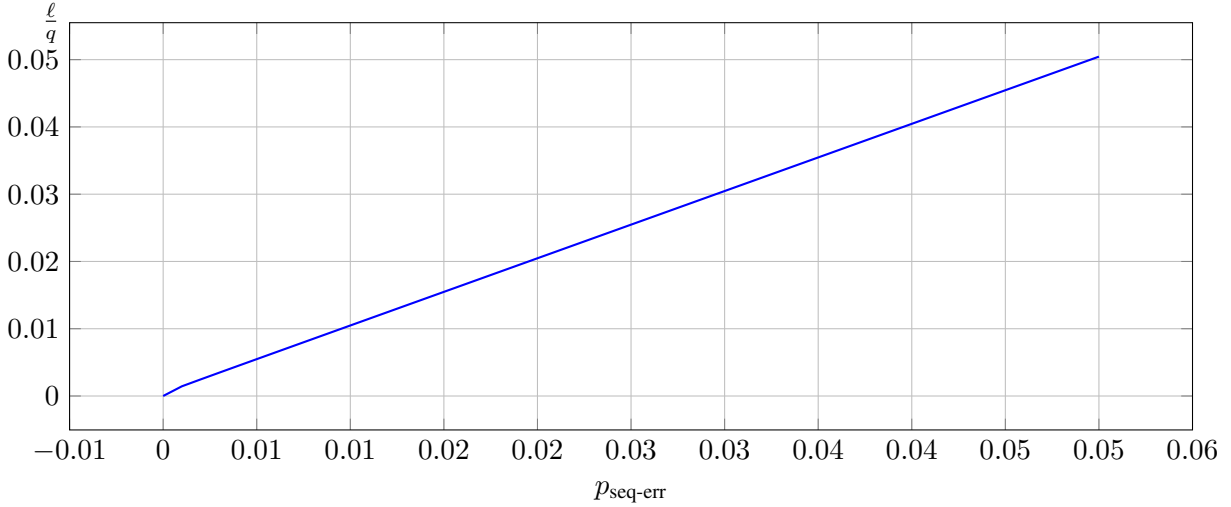
\begin{figure}[H]
    \centering

\begin{tikzpicture}
\begin{axis}[
    xlabel={$p_{\text{seq-err}}$},
    ylabel={$\frac{\ell}{q}$},
    ylabel style={at={(axis description cs:0,1)}, anchor=east, rotate=-90},
    width=\columnwidth, height=7cm,
    grid=major,
    ytick distance=0.01,
    legend pos=north east,
    legend cell align={left},
    scaled ticks=false,
    tick label style={/pgf/number format/fixed, /pgf/number format/precision=2},
  ]
  \addplot [blue, thick] table [x=p_e, y=t, col sep=comma] {images/decision_boundary.csv};
\end{axis}
\end{tikzpicture}
    
    \caption{Decision boundary for the maximum-likelihood rule. The plot shows the threshold on the fraction \(\ell/q\) of reads reporting an edit as a function of the sequencing error rate $p_{\text{seq-err}}$.}
    \label{fig:decision_boundary}
\end{figure}
Given a $q$, there is an $\ell^\star$ that marks the decision threshold. If the number of reads reporting an edit exceeds this threshold, the adversary decides that position $i$ is a key-site; otherwise, it decides that it is not a key-site.
Using this decision rule, the adversary has the following probability to make an error in determining the key-site status at a position $i$ of the challenge:
\begin{align*}
    p_{\text{err}}(q) 
    = 0.5 \Bigg( &\sum_{j=0}^{\ell^\star} \binom{q}{j} p_{\text{detect}}^{j} (1-p_{\text{detect}})^{q-j}
    \\ & + 
    \sum_{j= \ell^\star +1}^{q} \binom{q}{j} p_{\text{seq-err}}^{j} (1-p_{\text{seq-err}})^{q-j} \Bigg)
\end{align*}
where the first sum corresponds to the case where position $i$ is a key-site and the second sum to the case where it is not a key-site. 

If we let $X$ be the random variable representing the number of errors in the final challenge of the cloning game, then the success probability of the adversary is given by
\begin{align*}
    &\prob{\attack_{\adversary, \cfs{\profile}}(n) = 1}\\
    &\hspace{1cm} =\prob{X < t} 
    \\
    &\hspace{1cm} = \bincdf{t, n_{\text{chal}},p_{\text{err}} p_{\text{detect}} + (1 - p_{\text{err}}) (1 - p_{\text{detect}})}
\end{align*}
where $\bincdf{\cdot, \cdot, \cdot}$ is the binomial cumulative distribution function, $n_{\text{chal}}$ is the length of the challenge and $t$ is the error correcting capability of the code in the Fuzzy Extractor.

To analyze the unpredictability of the scheme, one would consider similar attack vectors as the coupon collector problem mentioned before. 

\section{Conclusions}\label{sec:conclusions}
This work introduced \acfp{cf} and the associated \acf{cfi} alongside properties relevant for security applications as a unified framework for security mechanisms realized by chemistry. This allows a systematic evaluation and comparison of different cryptographic schemes in the chemical domain. Building on the formalization of physically unclonable functions, we adapted concepts from \acp{puf} to the chemical domain and defined the components of a \ac{cfi}: generation and evaluation procedures, an extraction mechanism with helper data, and a verification interface operating on noisy challenge--response behavior. This abstraction connects laboratory procedures with cryptographic reasoning and enables meaningful comparisons across implementations.

Within this framework, we gave formal definitions of the central security properties—robustness, unclonability, and unpredictability—both in finite domains and asymptotically. We specified adversary capabilities and stated security goals via explicit security games that capture the adversary’s power and success criteria. This places chemically realized mechanisms on a footing comparable to established cryptographic primitives while reflecting laboratory constraints. Instead of quantifying the cost of a specific attack on the scheme, we focused on defining the capabilities of the adversary in the security games and relating the effort required to break the scheme to the effort of evaluating the scheme via a security parameter $n$. 

We instantiated the framework with \ac{dna}-based mechanisms, including operable random \ac{dna} (\ac{ordna}) and Genomic Sequence Encryption (\ac{gse}). These case studies show how the \ac{cfi} components and security games apply in practice, how chemical design and measurement models translate into quantitative guarantees for robustness and unclonability, and how standard extraction techniques yield reproducible verification. This led to the conclusion that \ac{dna}-based schemes might be suitable for enabling security applications. Beyond authentication of materials, we outlined how \acp{cf} enable key-generation workflows, including distributed settings where multiple near-identical instances reconstruct the same secret from noisy measurements.

By abstracting chemical mechanisms while retaining measurable interfaces, the framework provides a principled basis for designing, analyzing, and comparing chemically grounded authentication and key-generation schemes. We expect it to support a transition from ad hoc designs toward interoperable, quantitatively analyzable, and practically deployable systems.

\appendices

\section{Unclonability of orDNA}\label{sec:orDNA-unclonability}
The MinHash algorithm used in \cite{luescher_2024_ChemicalUnclonablea} approximates the weighted Jaccard similarity \cite{liu_2022_CMashFast}. Consequently, the probability that two vectors $y$ and $y'$ collide under MinHash is an unbiased estimate of the (weighted) Jaccard similarity of the corresponding $k$-mer sets \cite{liu_2022_CMashFast}. Let $K$ and $K'$ denote the indicator vectors of the $k$-mer multisets obtained from $\cf{\profile}(x)$ and $\cf{\profile'}(x)$, respectively. Then we have $K, K' \in \left[ 0, m \right]^{4^k} $ with $L_1$-norm $|| K' ||_1 = ||K||_1 = \sum_{i=1}^{4^k} K_i = m$. 
    Let us denote by
    \begin{align*}
        \jaccard{K, K'} = \frac{\sum_{i=1}^{4^k} \min(K_i, K_i')}{\sum_{i=1}^{4^k} \max(K_i, K_i')}
    \end{align*}
    their weighted Jaccard similarity. Define the random variable $W$ as the number of coordinates where $y$ and $y'$ agree, i.e., $W = n - \dist{y, y'}$. Then, we observe that
    \begin{align}
        \tau &= \prob{\attack_{\adversary, \cfs{\profile}}(n) = 1} \nonumber
        \\
        &= \prob{\cfs{\profile'}(x,h) = \cfs{\profile}(x,h)} \nonumber
        \\
        &= \prob{(z',h) = (z,h) } \nonumber
        \\
        &= \prob{z' = z } \nonumber
        \\
        &= \prob{\dist{y',y} \leq t } \nonumber
        \\
        &= \prob{W \geq n-t}, \label{eq:tau_W}
    \end{align}
    where $t$ is the error-correction capability of the Reed–Solomon code of length $n$ used in the extraction algorithm. By definition, the random variable $W$ follows a binomial distribution with parameters $n=255$ and $p_e = \E_{\pi \in S_{4^k}}\left[\jaccard{K, K'}\right]$.
    Since the $k$-mers are sampled from a DNA sequence according to an underlying distribution $\mathbf{p} = (p_1, \ldots, p_{4^k})$, we have that $K = (K_1, \ldots, K_{4^k})$ is a random variable following a multinomial distribution with parameters $m$ and $\mathbf{p}$. The same holds for $K' = (K'_1, \ldots, K'_{4^k})$ with distribution $\mathbf{p}'$. Assume $\mathbf{p}'$ is a permutation of $\mathbf{p}$, i.e., $\mathbf{p}' = (p_{\pi(1)}, \ldots, p_{\pi(4^k)})$ for some permutation $\pi \in S_{4^k}$.

    A standard upper bound relates the weighted Jaccard similarity to the Dice–Sørensen coefficient. Thus,
    \begin{align*}
        p_e 
        &= \E_{\pi \in S_{4^k}}\left[\jaccard{K, K'}\right]\\
        &\leq \E_{\pi \in S_{4^k}}\left[\frac{2 \sum_{i=1}^{4^k} \min(K_i, K_i')}{\sum_{i=1}^{4^k} K_i + \sum_{i=1}^{4^k} K_i'}\right]
        \\
        &= \E_{\pi \in S_{4^k}}\left[\frac{2 \sum_{i=1}^{4^k} \min(K_i, K_i')}{2m}\right]\\
        &= \frac{1}{m} \E_{\pi \in S_{4^k}}\left[\sum_{i=1}^{4^k} \min(K_i, K_i')\right].
    \end{align*}

    It is shown in Appendix \ref{sec:proof_schur_concavity} that $\E_{\pi \in S_{4^k}}\left[\sum_{i=1}^{4^k} \min(K_i, K_i')\right]$ is Schur-concave in $\mathbf{p}$. Hence, the maximum is achieved when $\mathbf{p}$ is uniformly distributed over its support as shown in \cite{marshall_2011_InequalitiesTheory}. Let $\tilde{K}, \tilde{K'}$ denote the corresponding uniform random variables over a support of size $s$ being the maximum number of unique $k$-mers that can appear in a challenge. Then, we have
    \begin{align*}
        p_e 
        & \leq \frac{1}{m}  \E_{\pi \in S_{4^k}}\left[\sum_{i=1}^{4^k} \min(K_i, K_i')\right]\\
        & \leq \frac{1}{m}  \E_{\pi \in S_{4^k}}\left[\sum_{i=1}^{4^k} \min(\tilde{K}_i, \tilde{K'}_i)\right].
    \end{align*}
    Under this simplification, $\min(\tilde{K}_i, \tilde{K'}_i)$ equals $1/s$ when $i$ lies in the intersection of the two supports (a collision) and $0$ otherwise. The expected value therefore reduces to the expected number of collisions under a uniformly random permutation $\pi$. For support size $s$, the probability of exactly $j$ collisions is
    \begin{align*}
        \prob{j \text{ collisions}} = \frac{\binom{4^k}{s} \binom{s}{j} \binom{4^k - s - j}{\,s-j\,}}{\binom{4^k}{s}^2}
        = \frac{\binom{s}{j} \binom{4^k - s - j}{\,s-j\,}}{\binom{4^k}{s}},
    \end{align*}
    whence
    \begin{align*}
        \E_{\pi \in S_{4^k}}\left[\sum_{i=1}^{4^k} \min(\tilde{K}_i, \tilde{K'}_i)\right] 
        &= \frac{1}{s} \sum_{j=0}^{s} j \prob{j \text{ collisions}} 
        \\
        &= \frac{1}{s} \sum_{j=0}^{s} j \frac{ \binom{s}{j} \binom{4^k - s - j}{s-j}}{\binom{4^k}{s}}.
    \end{align*}
    In \cite{luescher_2024_ChemicalUnclonablea}, the number of unique $k$-mers per readout for $k~=~8$ satisfies $s~<~200$. Hence, we obtain the following upper bound for $p_e$
    \begin{align*}
        p_e < \frac{1}{200} \sum_{i=0}^{200} i \, \frac{\binom{200}{i} \binom{4^8 - 200 - i}{\,200-i\,}}{\binom{4^8}{200}} < 0.0031.
    \end{align*}
    Substituting this bound into \Cref{eq:tau_W} yields
    \begin{align*}
        \tau &= \prob{W \geq n-t | x \notin \mathcal{Q}_{\adversary}}
        \\
        &\leq \sum_{i=144}^{255} \binom{255}{i} (0.0031)^i (0.9969)^{255-i}
        \\
        < 10^{-280}.
    \end{align*}

\section{Proof of Schur-Concavity of Expected Weighted Jaccard Similarity}\label{sec:proof_schur_concavity}
\begin{theorem}[Schur-concavity]\label{the:schur_concavity}
    Let $X_i$ and $Y_i$ denote the counts of the $i$-th $k$-mer in two profiles of length $s$ drawn i.i.d. from probability vectors $\mathbf{p}$ and $\mathbf{p}'$, respectively, where $\mathbf{p}'$ is a permutation of $\mathbf{p}$. 
    Define
    \begin{align*}
        F(\mathbf{p}) = \E_{\pi \in S_{4^k}}\Bigg[\sum_{i=1}^{4^k} \min\big(X_i, Y_i\big)\Bigg],
    \end{align*}
    where $S_{4^k}$ is the symmetric group on $4^k$ elements and the expectation is over a uniformly random permutation $\pi$. Then $F(\mathbf{p})$ is Schur-concave.
\end{theorem}

\begin{proof}
    To establish Schur-concavity of $F(\mathbf{p})$ with respect to the probability vector $\mathbf{p}$ over $k$-mers, it suffices to prove symmetry and concavity. 

    \subsection*{Symmetry}
    Let $\sigma$ be any permutation of the index set $\{1,2,\ldots,4^k\}$. Consider $\mathbf{p}_{\sigma} = (p_{\sigma(1)},\ldots,p_{\sigma(4^k)})$ and write
    \begin{align*}
        F(\mathbf{p}_{\sigma}) 
        = \E_{\pi \in S_{4^k}}\left[\sum_{i=1}^{4^k} \min\big(X_i^{\sigma}, Y_i^{\sigma}\big)\right],
    \end{align*}
    where $X^{\sigma}$ and $Y^{\sigma}$ are the multinomial count vectors obtained from $\mathbf{p}_{\sigma}$ and $\mathbf{p}'_{\sigma}$, respectively. Since $\mathbf{p}'$ is a permutation of $\mathbf{p}$, we have $\mathbf{p}'_{\sigma} = (\mathbf{p}_{\sigma})_{\pi} = \mathbf{p}_{\sigma\pi}$.

    The probability of observing a specific count vector $\mathbf{x}^{\sigma} = (x_1^{\sigma},\ldots,x_{4^k}^{\sigma})$ under $\mathbf{p}_{\sigma}$ is
    \begin{align*}
        \prob{X^{\sigma} = \mathbf{x}^{\sigma}} = \frac{s!}{x_1^{\sigma}! x_2^{\sigma}! \ldots x_{4^k}^{\sigma}!} \prod_{i=1}^{4^k} p_{\sigma(i)}^{x_i^{\sigma}}
    \end{align*}
    With the change of variables $x_i = x_{\sigma^{-1}(i)}^{\sigma}$ for all $i$, we obtain
    \begin{align*}
        \prod_{i=1}^{4^k} \big(p_{\sigma(i)}\big)^{x_{\sigma(i)}} 
        = \prod_{i=1}^{4^k} p_i^{\,x_i^{\sigma}}.
    \end{align*}
    Hence $X^{\sigma}$ has the same distribution as $X$; the same reasoning applies to $Y^{\sigma}$. The summand $\sum_{i} \min(X_i,Y_i)$ is invariant under index relabeling, and since $\{\sigma\pi:~\pi~\in~S_{4^k}\}=S_{4^k}$, the expectation is unchanged. Therefore, $F(\mathbf{p}_{\sigma})=F(\mathbf{p})$, i.e., $F$ is symmetric.

    \subsection*{Concavity}
    To show concavity, we must prove that for any probability vectors $\mathbf{p}$ and $\mathbf{q}$ over $k$-mers and any $\lambda \in [0,1]$,
    \begin{align*}
        F(\lambda \mathbf{p} + (1-\lambda) \mathbf{q}) 
        \;\geq\; \lambda F(\mathbf{p}) + (1-\lambda) F(\mathbf{q}).
    \end{align*}

    For a fixed permutation $\pi$, define $F_{\pi}(\mathbf{p}) ~=~\sum_{i=1}^{4^k}~\min(X_i, Y_i)$. Then
    \begin{align*}
        F(\mathbf{p}) 
        = \E_{\pi \in S_{4^k}}\big[F_{\pi}(\mathbf{p})\big] 
        = \frac{1}{(4^k)!} \sum_{\pi \in S_{4^k}} F_{\pi}(\mathbf{p}).
    \end{align*}
    Hence, it is enough to show that $F_{\pi}(\mathbf{p})$ is concave for each fixed $\pi$, since a positive linear combination of concave functions is concave.

    Let $X^{(\mathbf{p})}$ and $X^{(\mathbf{q})}$ be the multinomial count vectors corresponding to $\mathbf{p}$ and $\mathbf{q}$, respectively, and let $X^{\mathrm{comb}}$ correspond to $\lambda\mathbf{p} + (1-\lambda)\mathbf{q}$. Define $Y^{(\mathbf{p})}$, $Y^{(\mathbf{q})}$, and $Y^{\mathrm{comb}}$ analogously under the permuted probabilities.

    Observe that $X^{\mathrm{comb}}$ can be realized by a per-trial mixture: for each of the $s$ draws, with probability $\lambda$ the draw is from $\mathbf{p}$ and with probability $1-\lambda$ it is from $\mathbf{q}$. A similar construction holds for $Y^{\mathrm{comb}}$. Using the concavity of the function $\min(a,b)$ and summing over $i$, we obtain
    \begin{align*}
        \sum_{i=1}^{4^k} \min\big(X_i^{\mathrm{comb}}, Y_i^{\mathrm{comb}}\big)
        \;\geq\; 
        \lambda \sum_{i=1}^{4^k} \min\big(X_i^{(\mathbf{p})}, Y_i^{(\mathbf{p})}\big)
        \\
        + (1-\lambda) \sum_{i=1}^{4^k} \min\big(X_i^{(\mathbf{q})}, Y_i^{(\mathbf{q})}\big).
    \end{align*}
    Taking expectations preserves the inequality and yields the desired concavity for each fixed $\pi$, and hence for $F$.


\end{proof}

\section*{Acknowledgment}
This work was co-funded by the European Union (DiDAX, 101115134). Swiss participants in this project are supported by the Swiss State Secretariat for Education, Research and Innovation (SERI) under contract number 23.00330. The views and opinions expressed are those of the authors only and do not necessarily reflect those of the European Union, the European Research Council Executive Agency, or the Swiss State Secretariat for Education, Research and Innovation. Neither the European Union, the Swiss government, nor the granting authority can be held responsible for them.

\ifCLASSOPTIONcaptionsoff
  \newpage
\fi


\bibliographystyle{IEEEtranTMBMC}
\bibliography{DNA_authentication}

\begin{thebibliography}{10}
\baselineskip 12pt
\providecommand{\url}[1]{#1}
\csname url@samestyle\endcsname
\providecommand{\newblock}{\relax}
\providecommand{\bibinfo}[2]{#2}
\providecommand{\BIBentrySTDinterwordspacing}{\spaceskip=0pt\relax}
\providecommand{\BIBentryALTinterwordstretchfactor}{4}
\providecommand{\BIBentryALTinterwordspacing}{\spaceskip=\fontdimen2\font plus
\BIBentryALTinterwordstretchfactor\fontdimen3\font minus \fontdimen4\font\relax}
\providecommand{\BIBforeignlanguage}[2]{{%
\expandafter\ifx\csname l@#1\endcsname\relax
\typeout{** WARNING: IEEEtran.bst: No hyphenation pattern has been}%
\typeout{** loaded for the language `#1'. Using the pattern for}%
\typeout{** the default language instead.}%
\else
\language=\csname l@#1\endcsname
\fi
#2}}
\providecommand{\BIBdecl}{\relax}
\BIBdecl

\bibitem{luescher_2024_ChemicalUnclonablea}
\BIBentryALTinterwordspacing
A.~M. Luescher, A.~L. Gimpel, W.~J. Stark, R.~Heckel, and R.~N. Grass, ``Chemical unclonable functions based on operable random {{DNA}} pools,'' \emph{Nature Communications}, vol.~15, no.~1, p. 2955, Apr. 2024. [Online]. Available: \url{https://www.nature.com/articles/s41467-024-47187-7}
\BIBentrySTDinterwordspacing

\bibitem{volf_2023_CryptographyDNA}
\BIBentryALTinterwordspacing
V.~Volf, S.~Zhang, K.~M. Song, S.~Qian, F.~Chen, and G.~M. Church, ``Cryptography in the {{DNA}} of living cells enabled by multi-site base editing,'' p. 2023.11.15.567131, Nov. 2023. [Online]. Available: \url{https://www.biorxiv.org/content/10.1101/2023.11.15.567131v1}
\BIBentrySTDinterwordspacing

\bibitem{berezin_2024_CryptographicApproaches}
\BIBentryALTinterwordspacing
C.-T. Berezin, S.~Peccoud, D.~M. Kar, and J.~Peccoud, ``Cryptographic approaches to authenticating synthetic {{DNA}} sequences,'' \emph{Trends in Biotechnology}, Feb. 2024. [Online]. Available: \url{https://www.sciencedirect.com/science/article/pii/S0167779924000313}
\BIBentrySTDinterwordspacing

\bibitem{heider_2007_DNAbasedWatermarks}
\BIBentryALTinterwordspacing
D.~Heider and A.~Barnekow, ``{{DNA-based}} watermarks using the {{DNA-Crypt}} algorithm,'' \emph{BMC Bioinformatics}, vol.~8, no.~1, p. 176, May 2007. [Online]. Available: \url{https://doi.org/10.1186/1471-2105-8-176}
\BIBentrySTDinterwordspacing

\bibitem{heider_2008_DNAWatermarks}
\BIBentryALTinterwordspacing
------, ``{{DNA}} watermarks: {{A}} proof of concept,'' \emph{BMC Molecular Biology}, vol.~9, no.~1, p.~40, Apr. 2008. [Online]. Available: \url{https://doi.org/10.1186/1471-2199-9-40}
\BIBentrySTDinterwordspacing

\bibitem{jupiter_2010_DNAWatermarking}
\BIBentryALTinterwordspacing
D.~C. Jupiter, T.~A. Ficht, J.~Samuel, Q.-M. Qin, and P.~de~Figueiredo, ``{{DNA Watermarking}} of {{Infectious Agents}}: {{Progress}} and {{Prospects}},'' \emph{PLOS Pathogens}, vol.~6, no.~6, p. e1000950, Jun. 2010. [Online]. Available: \url{https://journals.plos.org/plospathogens/article?id=10.1371/journal.ppat.1000950}
\BIBentrySTDinterwordspacing

\bibitem{lee_2014_DWTBased}
\BIBentryALTinterwordspacing
S.-H. Lee, ``{{DWT}} based coding {{DNA}} watermarking for {{DNA}} copyright protection,'' \emph{Information Sciences}, vol. 273, pp. 263--286, Jul. 2014. [Online]. Available: \url{https://www.sciencedirect.com/science/article/pii/S0020025514003235}
\BIBentrySTDinterwordspacing

\bibitem{kar_2018_DigitalSignatures}
\BIBentryALTinterwordspacing
D.~M. Kar, I.~Ray, J.~Gallegos, and J.~Peccoud, ``Digital {{Signatures}} to {{Ensure}} the {{Authenticity}} and {{Integrity}} of {{Synthetic DNA Molecules}},'' in \emph{Proceedings of the {{New Security Paradigms Workshop}}}, ser. {{NSPW}} '18.\hskip 1em plus 0.5em minus 0.4em\relax New York, NY, USA: Association for Computing Machinery, Aug. 2018, pp. 110--122. [Online]. Available: \url{https://doi.org/10.1145/3285002.3285007}
\BIBentrySTDinterwordspacing

\bibitem{kar_2019_SystematizationKnowledge}
\BIBentryALTinterwordspacing
D.~M. Kar and I.~Ray, ``Systematization of {{Knowledge}} and {{Implementation}}: {{Short Identity-Based Signatures}},'' Aug. 2019. [Online]. Available: \url{http://arxiv.org/abs/1908.05366}
\BIBentrySTDinterwordspacing

\bibitem{kar_2019_ThatsMy}
------, ``That's {{My DNA}}: {{Detecting Malicious Tampering}} of {{Synthesized DNA}},'' in \emph{Data and {{Applications Security}} and {{Privacy XXXIII}}}, S.~N. Foley, Ed.\hskip 1em plus 0.5em minus 0.4em\relax Cham: Springer International Publishing, 2019, pp. 61--80.

\bibitem{kar_2020_SynthesizingDNA}
\BIBentryALTinterwordspacing
D.~M. Kar, I.~Ray, J.~Gallegos, J.~Peccoud, and I.~Ray, ``Synthesizing {{DNA}} molecules with identity-based digital signatures to prevent malicious tampering and enabling source attribution,'' \emph{Journal of Computer Security}, vol.~28, no.~4, pp. 437--467, Jan. 2020. [Online]. Available: \url{https://content.iospress.com/articles/journal-of-computer-security/jcs191383}
\BIBentrySTDinterwordspacing

\bibitem{clelland_1999_HidingMessages}
\BIBentryALTinterwordspacing
C.~T. Clelland, V.~Risca, and C.~Bancroft, ``Hiding messages in {{DNA}} microdots,'' \emph{Nature}, vol. 399, no. 6736, pp. 533--534, Jun. 1999. [Online]. Available: \url{https://www.nature.com/articles/21092}
\BIBentrySTDinterwordspacing

\bibitem{vippathalla_2023_SecureStorage}
\BIBentryALTinterwordspacing
P.~K. Vippathalla and N.~Kashyap, ``The {{Secure Storage Capacity}} of a {{DNA Wiretap Channel Model}},'' \emph{IEEE Transactions on Information Theory}, vol.~69, no.~9, pp. 5550--5569, Sep. 2023. [Online]. Available: \url{https://ieeexplore.ieee.org/document/10122215}
\BIBentrySTDinterwordspacing

\bibitem{leier_2000_CryptographyDNA}
\BIBentryALTinterwordspacing
A.~Leier, C.~Richter, W.~Banzhaf, and H.~Rauhe, ``Cryptography with {{DNA}} binary strands,'' \emph{Biosystems}, vol.~57, no.~1, pp. 13--22, Jun. 2000. [Online]. Available: \url{https://www.sciencedirect.com/science/article/pii/S0303264700000836}
\BIBentrySTDinterwordspacing

\bibitem{cui_2019_IncorporatingRandomness}
\BIBentryALTinterwordspacing
M.~Cui and Y.~Zhang, ``Incorporating {{Randomness}} into {{DNA Steganography}} to {{Realize Secondary Secret}} key, {{Self-destruction}}, and {{Quantum Key Distribution-like Function}},'' Aug. 2019. [Online]. Available: \url{http://biorxiv.org/lookup/doi/10.1101/725499}
\BIBentrySTDinterwordspacing

\bibitem{grass_2020_GenomicEncryption}
\BIBentryALTinterwordspacing
R.~N. Grass, R.~Heckel, C.~Dessimoz, and W.~J. Stark, ``Genomic {{Encryption}} of {{Digital Data Stored}} in {{Synthetic DNA}},'' \emph{Angewandte Chemie International Edition}, vol.~59, no.~22, pp. 8476--8480, 2020. [Online]. Available: \url{https://onlinelibrary.wiley.com/doi/abs/10.1002/anie.202001162}
\BIBentrySTDinterwordspacing

\bibitem{schaudy_2020_LDNADuplex}
\BIBentryALTinterwordspacing
E.~Schaudy, M.~M. Somoza, and J.~Lietard, ``L-{{DNA Duplex Formation}} as a {{Bioorthogonal Information Channel}} in {{Nucleic Acid-Based Surface Patterning}},'' \emph{Chemistry -- A European Journal}, vol.~26, no.~63, pp. 14\,310--14\,314, 2020. [Online]. Available: \url{https://onlinelibrary.wiley.com/doi/abs/10.1002/chem.202001871}
\BIBentrySTDinterwordspacing

\bibitem{gupta_2026_DNAData}
\BIBentryALTinterwordspacing
M.~K. Gupta, Z.~Ping, A.~Rasool, B.~Cao, and P.-Y. Burgi, Eds., \emph{{{DNA Data Storage}}: {{Current Approaches}} and {{Emerging Trends}}}.\hskip 1em plus 0.5em minus 0.4em\relax Singapore: Springer Nature Singapore, 2026. [Online]. Available: \url{https://link.springer.com/10.1007/978-981-95-9450-4}
\BIBentrySTDinterwordspacing

\bibitem{limor_2026_StochasticPerformance}
\BIBentryALTinterwordspacing
O.~Limor, R.~Shafir, and Z.~Yakhini, ``A {{Stochastic Performance Analysis Framework}} for {{DNA-Based Authentication Protocols}},'' \emph{IEEE Transactions on Molecular, Biological, and Multi-Scale Communications}, vol.~12, pp. 937--946, 2026. [Online]. Available: \url{https://ieeexplore.ieee.org/document/11660756/}
\BIBentrySTDinterwordspacing

\bibitem{katz_2020_IntroductionModern}
\BIBentryALTinterwordspacing
J.~Katz and Y.~Lindell, \emph{Introduction to {{Modern Cryptography}}}, 3rd~ed.\hskip 1em plus 0.5em minus 0.4em\relax {Chapman and Hall/CRC}, Dec. 2020. [Online]. Available: \url{https://www.taylorfrancis.com/books/9781351133029}
\BIBentrySTDinterwordspacing

\bibitem{nationalinstituteofstandardsandtechnologyus_2016_SubmissionRequirements}
\BIBentryALTinterwordspacing
{National Institute of Standards {and} Technology (US)}, ``Submission requirements and evaluation criteria for the post-quantum cryptography standardization process,'' 2016. [Online]. Available: \url{https://csrc.nist.gov/csrc/media/projects/post-quantum-cryptography/documents/call-for-proposals-final-dec-2016.pdf}
\BIBentrySTDinterwordspacing

\bibitem{armknecht_2010_MemoryLeakageResilient}
\BIBentryALTinterwordspacing
F.~Armknecht, R.~Maes, A.-R. Sadeghi, B.~Sunar, and P.~Tuyls, ``Memory {{Leakage-Resilient Encryption Based}} on {{Physically Unclonable Functions}},'' in \emph{Towards {{Hardware-Intrinsic Security}}}, A.-R. Sadeghi and D.~Naccache, Eds.\hskip 1em plus 0.5em minus 0.4em\relax Berlin, Heidelberg: Springer Berlin Heidelberg, 2010, pp. 135--164. [Online]. Available: \url{http://link.springer.com/10.1007/978-3-642-14452-3_6}
\BIBentrySTDinterwordspacing

\bibitem{skoric_2005_RobustKey}
\BIBentryALTinterwordspacing
B.~{\v S}kori{\'c}, P.~Tuyls, and W.~Ophey, ``Robust {{Key Extraction}} from {{Physical Uncloneable Functions}},'' in \emph{Applied {{Cryptography}} and {{Network Security}}}, D.~Hutchison, T.~Kanade, J.~Kittler, J.~M. Kleinberg, F.~Mattern, J.~C. Mitchell, M.~Naor, O.~Nierstrasz, C.~Pandu~Rangan, B.~Steffen, M.~Sudan, D.~Terzopoulos, D.~Tygar, M.~Y. Vardi, G.~Weikum, J.~Ioannidis, A.~Keromytis, and M.~Yung, Eds.\hskip 1em plus 0.5em minus 0.4em\relax Berlin, Heidelberg: Springer Berlin Heidelberg, 2005, vol. 3531, pp. 407--422. [Online]. Available: \url{http://link.springer.com/10.1007/11496137_28}
\BIBentrySTDinterwordspacing

\bibitem{lim_2005_ExtractingSecret}
\BIBentryALTinterwordspacing
D.~Lim, J.~Lee, B.~Gassend, G.~Suh, M.~Van~Dijk, and S.~Devadas, ``Extracting secret keys from integrated circuits,'' \emph{IEEE Transactions on Very Large Scale Integration (VLSI) Systems}, vol.~13, no.~10, pp. 1200--1205, Oct. 2005. [Online]. Available: \url{http://ieeexplore.ieee.org/document/1561249/}
\BIBentrySTDinterwordspacing

\bibitem{pappu_2002_PhysicalOneWay}
\BIBentryALTinterwordspacing
R.~Pappu, B.~Recht, J.~Taylor, and N.~Gershenfeld, ``Physical {{One-Way Functions}},'' \emph{Science}, vol. 297, no. 5589, pp. 2026--2030, Sep. 2002. [Online]. Available: \url{https://www.science.org/doi/abs/10.1126/science.1074376}
\BIBentrySTDinterwordspacing

\bibitem{gassend_2002_SiliconPhysical}
\BIBentryALTinterwordspacing
B.~Gassend, D.~Clarke, M.~Van~Dijk, and S.~Devadas, ``Silicon physical random functions,'' in \emph{Proceedings of the 9th {{ACM}} Conference on {{Computer}} and Communications Security}.\hskip 1em plus 0.5em minus 0.4em\relax Washington, DC USA: ACM, Nov. 2002, pp. 148--160. [Online]. Available: \url{https://dl.acm.org/doi/10.1145/586110.586132}
\BIBentrySTDinterwordspacing

\bibitem{katzenbeisser_2011_RecyclablePUFs}
\BIBentryALTinterwordspacing
S.~Katzenbeisser, {\"U}.~Kocaba{\c s}, V.~Van Der~Leest, A.-R. Sadeghi, G.-J. Schrijen, and C.~Wachsmann, ``Recyclable {{PUFs}}: Logically reconfigurable {{PUFs}},'' \emph{Journal of Cryptographic Engineering}, vol.~1, no.~3, pp. 177--186, Nov. 2011. [Online]. Available: \url{http://link.springer.com/10.1007/s13389-011-0016-9}
\BIBentrySTDinterwordspacing

\bibitem{vandijk_2012_PhysicalUnclonable}
\BIBentryALTinterwordspacing
M.~{van Dijk} and U.~R{\"u}hrmair, ``Physical unclonable functions in cryptographic protocols: {{Security}} proofs and impossibility results,'' Cryptology ePrint Archive, Paper 2012/228, 2012. [Online]. Available: \url{https://eprint.iacr.org/2012/228}
\BIBentrySTDinterwordspacing

\bibitem{armknecht_2011_FormalizationSecurity}
\BIBentryALTinterwordspacing
F.~Armknecht, R.~Maes, A.-R. Sadeghi, F.-X. Standaert, and C.~Wachsmann, ``A {{Formalization}} of the {{Security Features}} of {{Physical Functions}},'' in \emph{2011 {{IEEE Symposium}} on {{Security}} and {{Privacy}}}.\hskip 1em plus 0.5em minus 0.4em\relax Oakland, CA, USA: IEEE, May 2011, pp. 397--412. [Online]. Available: \url{http://ieeexplore.ieee.org/document/5958042/}
\BIBentrySTDinterwordspacing

\bibitem{li_2022_GeneticPhysical}
\BIBentryALTinterwordspacing
Y.~Li, M.~M. Bidmeshki, T.~Kang, C.~M. Nowak, Y.~Makris, and L.~Bleris, ``Genetic physical unclonable functions in human cells,'' \emph{Science Advances}, vol.~8, no.~18, p. eabm4106, May 2022. [Online]. Available: \url{https://www.science.org/doi/10.1126/sciadv.abm4106}
\BIBentrySTDinterwordspacing

\bibitem{ruhrmair_2009_FoundationsPhysical}
\BIBentryALTinterwordspacing
U.~R{\"u}hrmair, J.~S{\"o}lter, and F.~Sehnke, ``On the foundations of physical unclonable functions,'' Cryptology ePrint Archive, Paper 2009/277, 2009. [Online]. Available: \url{https://eprint.iacr.org/2009/277}
\BIBentrySTDinterwordspacing

\bibitem{ruhrmair_2010_ModelingAttacks}
\BIBentryALTinterwordspacing
U.~R{\"u}hrmair, F.~Sehnke, J.~S{\"o}lter, G.~Dror, S.~Devadas, and J.~Schmidhuber, ``Modeling attacks on physical unclonable functions,'' in \emph{Proceedings of the 17th {{ACM}} Conference on {{Computer}} and Communications Security}.\hskip 1em plus 0.5em minus 0.4em\relax Chicago Illinois USA: ACM, Oct. 2010, pp. 237--249. [Online]. Available: \url{https://dl.acm.org/doi/10.1145/1866307.1866335}
\BIBentrySTDinterwordspacing

\bibitem{woidasky_2020_InorganicFluorescent}
\BIBentryALTinterwordspacing
J.~Woidasky, I.~Sander, A.~Schau, J.~Moesslein, P.~Wendler, D.~Wacker, G.~Gao, D.~Kirchenbauer, V.~Kumar, D.~Busko, I.~A. Howard, B.~S. Richards, A.~Turshatov, S.~Wiethoff, and C.~{Lang-Koetz}, ``Inorganic fluorescent marker materials for identification of post-consumer plastic packaging,'' \emph{Resources, Conservation and Recycling}, vol. 161, p. 104976, Oct. 2020. [Online]. Available: \url{https://linkinghub.elsevier.com/retrieve/pii/S0921344920302949}
\BIBentrySTDinterwordspacing

\bibitem{yang_2024_DigitalBarcodes}
\BIBentryALTinterwordspacing
Z.~Yang, J.~Chen, Y.~Xiao, C.~Yang, C.-X. Zhao, D.~Chen, and D.~A. Weitz, ``Digital {{Barcodes}} for {{High-Throughput Screening}},'' \emph{Chem \& Bio Engineering}, vol.~1, no.~1, pp. 2--12, Feb. 2024. [Online]. Available: \url{https://pubs.acs.org/doi/10.1021/cbe.3c00085}
\BIBentrySTDinterwordspacing

\bibitem{dodis_2004_FuzzyExtractors}
\BIBentryALTinterwordspacing
Y.~Dodis, L.~Reyzin, and A.~Smith, ``Fuzzy {{Extractors}}: {{How}} to {{Generate Strong Keys}} from {{Biometrics}} and {{Other Noisy Data}},'' in \emph{Advances in {{Cryptology}} - {{EUROCRYPT}} 2004}, T.~Kanade, J.~Kittler, J.~M. Kleinberg, F.~Mattern, J.~C. Mitchell, O.~Nierstrasz, C.~Pandu~Rangan, B.~Steffen, D.~Terzopoulos, D.~Tygar, M.~Y. Vardi, C.~Cachin, and J.~L. Camenisch, Eds.\hskip 1em plus 0.5em minus 0.4em\relax Berlin, Heidelberg: Springer Berlin Heidelberg, 2004, vol. 3027, pp. 523--540. [Online]. Available: \url{http://link.springer.com/10.1007/978-3-540-24676-3_31}
\BIBentrySTDinterwordspacing

\bibitem{liu_2022_CMashFast}
\BIBentryALTinterwordspacing
S.~Liu and D.~Koslicki, ``{{CMash}}: Fast, multi-resolution estimation of k-mer-based {{Jaccard}} and containment indices,'' \emph{Bioinformatics}, vol.~38, no. Supplement\_1, pp. i28--i35, Jun. 2022. [Online]. Available: \url{https://academic.oup.com/bioinformatics/article/38/Supplement_1/i28/6617499}
\BIBentrySTDinterwordspacing

\bibitem{broder_1998_ResemblanceContainment}
\BIBentryALTinterwordspacing
A.~Broder, ``On the resemblance and containment of documents,'' in \emph{Proceedings. {{Compression}} and {{Complexity}} of {{SEQUENCES}} 1997 ({{Cat}}. {{No}}.{{97TB100171}})}.\hskip 1em plus 0.5em minus 0.4em\relax Salerno, Italy: IEEE Comput. Soc, 1998, pp. 21--29. [Online]. Available: \url{http://ieeexplore.ieee.org/document/666900/}
\BIBentrySTDinterwordspacing

\bibitem{haveliwala_2000_ScalableTechniques}
T.~Haveliwala, A.~Gionis, and P.~Indyk, ``Scalable techniques for clustering the web,'' in \emph{Proc. of the {{Third International Workshop}} on the {{Web}} and {{Databases}} ({{Informal Proceedings}})}, vol. 129, 2000, p. 134.

\bibitem{juels_2006_FuzzyVault}
\BIBentryALTinterwordspacing
A.~Juels and M.~Sudan, ``A {{Fuzzy Vault Scheme}},'' \emph{Designs, Codes and Cryptography}, vol.~38, no.~2, pp. 237--257, Feb. 2006. [Online]. Available: \url{http://link.springer.com/10.1007/s10623-005-6343-z}
\BIBentrySTDinterwordspacing

\bibitem{compeau_2018_BioinformaticsAlgorithms}
P.~Compeau and P.~Pevzner, \emph{Bioinformatics Algorithms: An Active Learning Approach}, 3rd~ed.\hskip 1em plus 0.5em minus 0.4em\relax San Diego, CA: AL, Active Learning Publishers, 2018.

\bibitem{leskovec_2014_MiningMassive}
\BIBentryALTinterwordspacing
J.~Leskovec, A.~Rajaraman, and J.~D. Ullman, \emph{Mining of {{Massive Datasets}}}, 2nd~ed.\hskip 1em plus 0.5em minus 0.4em\relax Cambridge University Press, Nov. 2014. [Online]. Available: \url{https://www.cambridge.org/core/product/identifier/9781139924801/type/book}
\BIBentrySTDinterwordspacing

\bibitem{meiser_2020_DNASynthesis}
\BIBentryALTinterwordspacing
L.~C. Meiser, J.~Koch, P.~L. Antkowiak, W.~J. Stark, R.~Heckel, and R.~N. Grass, ``{{DNA}} synthesis for true random number generation,'' \emph{Nature Communications}, vol.~11, no.~1, p. 5869, Nov. 2020. [Online]. Available: \url{https://www.nature.com/articles/s41467-020-19757-y}
\BIBentrySTDinterwordspacing

\bibitem{alyanbaawi_2025_MITMDoSResistant}
\BIBentryALTinterwordspacing
A.~Alyanbaawi, ``{{MITM-}} and {{DoS-Resistant PUF Authentication}} for {{Industrial WSNs}} via {{Sensor-Initiated Registration}},'' \emph{Computers}, vol.~14, no.~9, p. 347, Aug. 2025. [Online]. Available: \url{https://www.mdpi.com/2073-431X/14/9/347}
\BIBentrySTDinterwordspacing

\bibitem{lounis_2023_LessonsLearned}
\BIBentryALTinterwordspacing
K.~Lounis and M.~Zulkernine, ``Lessons {{Learned}}: {{Analysis}} of {{PUF-based Authentication Protocols}} for {{IoT}},'' \emph{Digital Threats: Research and Practice}, vol.~4, no.~2, pp. 1--33, Jun. 2023. [Online]. Available: \url{https://dl.acm.org/doi/10.1145/3487060}
\BIBentrySTDinterwordspacing

\bibitem{grass_2015_RobustChemical}
\BIBentryALTinterwordspacing
R.~N. Grass, R.~Heckel, M.~Puddu, D.~Paunescu, and W.~J. Stark, ``Robust {{Chemical Preservation}} of {{Digital Information}} on {{DNA}} in {{Silica}} with {{Error}}-{{Correcting Codes}},'' \emph{Angewandte Chemie International Edition}, vol.~54, no.~8, pp. 2552--2555, Feb. 2015. [Online]. Available: \url{https://onlinelibrary.wiley.com/doi/10.1002/anie.201411378}
\BIBentrySTDinterwordspacing

\bibitem{kuzdralinski_2023_UnlockingPotential}
\BIBentryALTinterwordspacing
A.~Kuzdrali{\'n}ski, M.~Mi{\'s}kiewicz, H.~Szczerba, W.~Mazurczyk, J.~Nivala, and B.~Ksie{\.z}opolski, ``Unlocking the potential of {{DNA-based}} tagging: Current market solutions and expanding horizons,'' \emph{Nature Communications}, vol.~14, no.~1, p. 6052, Sep. 2023. [Online]. Available: \url{https://www.nature.com/articles/s41467-023-41728-2}
\BIBentrySTDinterwordspacing

\bibitem{didax-project_2024_StateArt}
\BIBentryALTinterwordspacing
{DiDAX-Project}, ``State of the art and commercial needs for authentication and in-product documentation,'' Nov. 2024. [Online]. Available: \url{https://didax.eu/wp-content/uploads/2025/04/Authentication-market-needs.pdf}
\BIBentrySTDinterwordspacing

\bibitem{marshall_2011_InequalitiesTheory}
A.~W. Marshall, I.~Olkin, and B.~C. Arnold, \emph{Inequalities: Theory of Majorization and Its Applications}, 2nd~ed.\hskip 1em plus 0.5em minus 0.4em\relax New York: Springer Science+Business Media, LLC, 2011.

\end{thebibliography}

\end{document}